\documentclass[11pt]{article}

\usepackage{amsmath,amssymb,amsthm}
\usepackage{algpseudocode,algorithm}
\usepackage{fullpage, xcolor}
\usepackage{bbm}
\usepackage{footnote}
\usepackage{graphicx}

\makesavenoteenv{tabular}
\makesavenoteenv{table}

\DeclareMathOperator*{\EE}{\mathbb{E}}

\DeclareMathOperator{\supp}{supp}

\newtheorem{theorem}{Theorem}[section]
\newtheorem{lemma}[theorem]{Lemma}
\newtheorem{claim}[theorem]{Claim}
\newtheorem{corollary}[theorem]{Corollary}
\newtheorem{fact}[theorem]{Fact}
\newtheorem{definition}[theorem]{Definition}

\makeatletter
\newenvironment{auxAlg}[1][htb]{%
	\renewcommand{\ALG@name}{AuxilliaryAlgorithm}
	\begin{algorithm}[#1]%
	}{\end{algorithm}}
\makeatother
\providecommand{\auxAlgRef}{AuxiliaryAlgorithm}

\newcommand{\eps}{\epsilon}

\providecommand{\logtwo}{\operatorname{\log}}
\providecommand{\olog}{\overline{\log}}
\providecommand{\oln}{\overline{\ln}}
\providecommand{\algone}{Algorithm~1}
\providecommand{\algtwo}{Algorithm~2}

\providecommand{\Current}{\mathit{Current}}
\providecommand{\LastVerified}{\mathit{LastVerified}}
\providecommand{\Candidate}{\mathit{Candidate}}
\providecommand{\VerificationPath}{\mathit{VerificationPath}}
\providecommand{\theroot}{\text{root}}
\providecommand{\depth}{\text{depth}}
\providecommand{\rightchild}{\text{right-child}}
\providecommand{\leftchild}{\text{left-child}}
\providecommand{\Live}{\mathit{Live}}
\providecommand{\parent}{\mathrm{parent}}
\providecommand{\Suspicious}{\mathit{Suspicious}}
\providecommand{\None}{\it{None}}
\providecommand{\SuspiciousCandidate}{\it{SuspiciousCandidate}}

\providecommand{\Hloglog}{H_2}
\providecommand{\Hlogloglog}{H_3}

\title{The entropy of lies:  playing twenty questions with a liar}

\author{Yuval Dagan\thanks{Department of Electrical Engineering and Computer Science, Massachusetts Institute of Technology.} \and Yuval Filmus\thanks{Computer Science Department, Technion. Taub Fellow --- supported by the Taub Foundations. The research was funded by ISF grant 1337/16.} \and Daniel Kane\thanks{Department of Computer Science and Engineering and Department of Mathematics, University of California at San Diego.} \and Shay Moran\thanks{Department of Computer Science, Princeton University.} \footnotemark[1]}

\begin{document}

\maketitle

\begin{abstract}

``Twenty questions'' is a  guessing game played by two players:
	Bob thinks of an integer between $1$ and $n$, 
	and Alice's goal is to recover it using a minimal number of Yes/No questions. 
	Shannon's entropy has a natural interpretation in this context.
	It characterizes the average number of questions used by an optimal strategy in the distributional variant of the game:
	let~$\mu$ be a distribution over~$[n]$, then the average number of questions used by an optimal strategy
	that recovers~$x\sim \mu$ is between $H(\mu)$ and $H(\mu)+1$.

We consider an extension of this game where at most $k$ questions can be answered falsely. 
	We extend the classical result by showing that an optimal strategy uses roughly 
	$H(\mu) + k\Hloglog(\mu)$ questions,
	where~$\Hloglog(\mu) = \sum_x \mu(x)\log\log\frac{1}{\mu(x)}$.
This also generalizes a result by Rivest et al.~(1980) for the uniform distribution.

Moreover, we design near optimal strategies that only use comparison queries of the form~``$x \leq c$?'' for $c\in[n]$.
	The usage of comparison queries lends itself naturally to the context of sorting, 
	where we derive sorting algorithms in the presence of adversarial noise.
\end{abstract}

\section{Introduction} \label{sec:introduction}

The ``twenty questions'' game is a cooperative game between two players: 
	Bob thinks of an integer between $1$ and $n$, 
	and Alice's goal is to recover it using the minimal number of Yes/No questions. 
	An optimal strategy for Alice is to perform binary search, using $\log n$ queries in the worst case.

The game becomes more interesting when Bob chooses his number according to a distribution~$\mu$ known to both players, 
	and Alice attempts to minimize the \emph{expected} number of questions. 
	In this case, the optimal strategy is to use a Huffman code for $\mu$, at an expected cost of roughly $H(\mu)$.

What happens when Bob is allowed to lie (either out of spite, or due to difficulties in the communication channel)? 
	R\'enyi~\cite{Renyi} and Ulam~\cite{Ulam} suggested a variant of the (non-distributional) ``twenty questions'' game, 
	in which Bob is allowed to lie $k$ times. 
	Rivest et al.~\cite{Rivest}, using ideas of Berlekamp~\cite{Berlekamp}, 
	showed that the optimal number of questions in this setting is roughly $\log n + k \log \log n$. 
	There are many other ways of allowing Bob to lie, some of which are described by Spencer and Winkler~\cite{SpencerWinkler} in their charming work, 
	and many others by Pelc~\cite{Pelc} in his comprehensive survey on the topic.

\paragraph{Distributional ``twenty questions'' with lies.}
This work addresses the distributional ``twenty questions'' game in the presence of lies. 
	In this setting, Bob draws an element $x$ according to a distribution $\mu$, 
	and Alice's goal is to recover the element using as few Yes/No questions as possible on average.
	The twist is that Bob, who knows Alice's strategy, is allowed to lie up to $k$ times. 
	Both Alice and Bob are allowed to use randomized strategies, and the average is measured according to both $\mu$ and the randomness of both parties.

Our main result shows that the expected number of questions in this case is 
\[
H(\mu) + k \Hloglog(\mu), \quad \text{ where } \Hloglog(\mu) = \sum_x \mu(x) \log \log \frac{1}{\mu(x)},
\]
up to an additive factor of $O(k \log k + k H_3(\mu))$, where $H_3(\mu) = \sum_x \mu(x) \log \log \log (1/\mu(x)))$
(here $\mu(x)$ is the probability of $x$ under $\mu$.) 
See Section~\ref{sec:mr} for a complete statement of this result.

When $\mu$ is the uniform distribution, 
	the expected number of queries that our algorithm makes is roughly $\log n + k \log \log n$, 
	matching the performance of the algorithm of Rivest et al. 
	However, the approach by Rivest et al.\ is tailored to their setting, 
	and the distributional setting requires new ideas.

As in the work of Rivest et al., our algorithms use only \emph{comparison queries}, 
	which are queries of the form ``$x \prec c$?'' (for some fixed value $c$).
	Moreover, our algoritms are efficient, requiring $O(n)$ preprocessing time and $O(\log n)$ time per question. 
	Our lower bounds, in contrast, apply to \emph{arbitrary} Yes/No queries.

\paragraph{Noisy sorting.}
One can apply binary search algorithms to implement insertion sort. 
	While sorting an array typically requires $\Theta(n \log n)$ \emph{sorting queries} of the form ``$x_i \prec x_j$?'', 
	there are situations where one has some prior knowledge about the correct ordering. 
	This may happen, for example, when maintaining a sorted array: one has to perform consecutive sorts, 
	where each sort is not expected to considerably change the locations of the elements. 
	Assuming a distribution $\Pi$ over the $n!$ possible permutations, 
	Moran and Yehudayoff~\cite{MY13} showed that sorting a $\Pi$-distributed array requires $H(\Pi) + O(n)$ sorting queries on average. 
	We extend this result to the case in which the answerer is allowed to lie $k$ times, giving an algorithm which uses the following expected number of queries:\footnote{Strictly speaking, this bound holds only under the mild condition that $k$ is at most exponential in $n$.}
%
\[
 H(\Pi) + O(nk).
\]
This result is tight, and matches the optimal algorithms for the uniform distribution due to Bagchi~\cite{Bagchi} and Long~\cite{Long}, which use $n\log n + O(nk)$ queries.

Table~\ref{tbl:query-comp} summarizes the query complexities of resilient and non-resilient searching and sorting algorithms, 
	in both the deterministic and the distributional settings. 
	To the best of our knowledge, we present the first resilient algorithms in the distributional setting.

\begin{table}[]
	\centering
	\begin{tabular}{|l|l|l|}
		\hline \textbf{Setting} & \textbf{Searching}           & \textbf{Sorting}                 \\
		\hline No lies; deterministic   & $\log n$ [classical]                     & $n \log n$ [classical]                       \\
		No lies; distributional  & $H(\mu)$ [classical]      & $H(\Pi) + O(n)$ \cite{MY13}                 \\
		$k$ lies; deterministic  & $\log n + k \log \log n$  \cite{Rivest} & $n \log n + \Theta(nk)$ \cite{Bagchi,Long,LRG91} \\
		$k$ lies; distributional & $H(\mu) + kH_2(\mu)$ [this paper] & $H(\Pi) + \Theta(nk)$ [this paper] \\ \hline
	\end{tabular}
\caption{Query complexities of searching and sorting in different settings, ignoring lower-order terms. All terms are exact upper and lower bounds except for those inside the $O(\cdot)$ and $\Theta(\cdot)$ notations.}
\label{tbl:query-comp}
\end{table}

%
%

\paragraph{On randomness.} 
All algorithms presented in the paper are randomized. 
	Since they only employ public randomness which is known for both players, 
	there exists a fixing of the randomness which yields a deterministic algorithm with the same (or possibly smaller) expected number of queries. 
	However, this comes at the cost of possibly increasing the running time of the algorithm (since we need to find a good fixing of the randomness);
	it would be interesting to derive an explicit efficient deterministic algorithm with a similar running time.

%



\subsection{Main ideas}


\paragraph{Upper bound.}
	Before presenting the ideas behind our algorithms, we explore several other ideas which give suboptimal results.
	The first approach that comes to mind is simulating the optimal non-resilient strategy, asking each question $2k+1$ times and taking the majority vote, which results in an algorithm using $\Theta(kH(\mu))$ queries on average.
	
	A better approach is using \emph{tree codes}, suggested by Schulman~\cite{schulman1996coding} as an approach for making interactive communication resilient to errors~\cite{gelles2017coding,schulman1996coding,kol2013interactive}. Tree codes are designed for a different error model, in which we are bounding the \emph{fraction} of lies rather than their absolute number; for an $\varepsilon$-fraction of lies, the best known constructions suffer a multiplicative overhead of $1 + O(\sqrt{\varepsilon})$~\cite{haeupler2014interactive}. In contrast, we are aiming at an \emph{additive} overhead of $kH_2(\mu)$.
	
	Using a packing bound, one can prove that there exists a (non-interactive) code of expected length roughly $H(\mu) + 2k H_2(\mu)$, coming much closer to the bound that we are able to get (but off by a factor of~2 from our target $H(\mu) + kH_2(\mu)$). The idea, which is similar to the proof of the Gilbert--Varshamov bound, is to construct a prefix code $w_1,\ldots,w_n$ in which the prefixes of~$w_i,w_j$ of length $\min(|w_i|,|w_j|)$ are at distance at least $2k+1$ (whence the factor $2k$ in the resulting bound); this can be done greedily. Apart from the inferior bound, two other disadvantages of this approach is that it is not efficient and uses arbitrary queries.
	
	In contrast to these prior techniques, which do not achieve the optimal complexity, might ask arbitrary questions, and could result in strategies which cannot be implemented efficiently, in this paper we design an efficient and nearly optimal strategy, relying on comparison queries only, and utilizing simple observations on the behavior of binary search trees under the presence of lies.
	
	\medskip

	Following the footsteps of Rivest et al.~\cite{Rivest}, 
	our upper bound is based on a binary search algorithm on the unit interval $[0,1]$, first suggested in this context by Gilbert and Moore~\cite{GilbertMoore}: 
	given~$x \in [0,1]$, the algorithm locates $x$ by first asking ``$x < 1/2$?''; depending on the answer, asking ``$x < 1/4$?'' or ``$x < 3/4$?''; and so on. 
	If $x \in [0,1]$ is chosen uniformly at random then the answers behave like an infinite sequence of random and uniform coin tosses.
	
\begin{figure}
	\centering
	\includegraphics{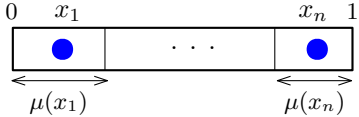}
	\caption{Representing items as centers of segments partitioning the interval $[0,1]$.}
	\label{fig:intro-gilbert-moore}
\end{figure}

	In order to apply this kind of binary search to the problem of identifying an unknown element (assuming truthful answers), we partition the unit interval $[0,1]$ into segments of lengths $\mu(x_1),\ldots,\mu(x_n)$, and label the center of each segment with the corresponding item (see Figure~\ref{fig:intro-gilbert-moore}). We then perform binary search until the current interval contains a single item. (In the proof, we use a slightly more sophisticated randomized placement of points which guarantees that the answers on \emph{each} element behave like an infinite sequence of random and uniform coin tosses.) 

\begin{figure} 
	\centering
	\includegraphics{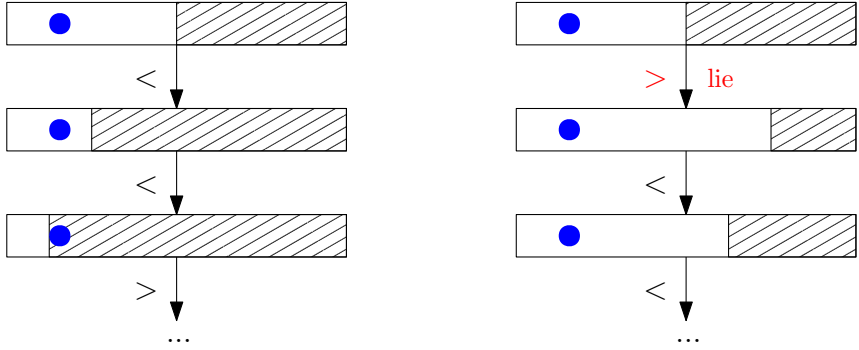}
	\caption{On the left, the operation of the algorithm without any lies. On the right, answerer lied on the first question. As a result, all future truthful answers are the same.}
	\label{fig:intro-main-obs}
\end{figure}

The main observation is that if a question ``$x < a$?'' is answered with a lie,
	this will be strongly reflected in subsequent answers (see Figure~\ref{fig:intro-main-obs}). Indeed, suppose that~$x < a$, but Bob claimed that~$x > a$. 
	All subsequent questions will be of the form ``$x < b?$'' for various $b > a$, 
	the truthful answer to all of which is $x < b$. 
	An observer taking notes of the proceedings will thus observe the following pattern of answers: 
	$>$ (the lie) followed by many $<$'s (possibly interspersed with up to $k-1$ many $>$'s, due to further lies). 
	This is suspicious since it is highly unlikely to obtain many $<$ answers in a row (the probability of getting $r$ such answers is just $2^{-r}$).

This suggests the following approach: for each question we will maintain a ``confidence interval'' consisting of $r(d)$ further questions (where $d$ is the index of the question). 
	At the end of the interval, we will check whether the situation is suspicious (as described in the preceding paragraph), 
	and if so, will ascertain by brute force the correct answer to the original question (by taking a majority of $2k+1$ answers), and restart the algorithm from that point.

The best choice for $r(d)$ turns out to be roughly $\log d$. 
	Each time Bob lies, our unrolling of the confidence interval results in a loss of $r(d)$ questions. 
	Since an item $x$ requires roughly $\log (1/\mu(x))$ questions to be discovered, 
	the algorithm has an overhead of roughly~$k r(\log (1/\mu(x))) \approx k \log \log (1/\mu(x))$ questions on element $x$, 
	resulting in an expected overhead of roughly $k H_2(\mu)$.
	
When implementing the algorithm, apart from the initial $O(n)$ time needed to setup the partition of $[0,1]$ into segments, the costliest step is to convert the intervals encountered in the binary search to comparison queries. This can be done in $O(\log n)$ time per query.


\paragraph{Lower bound.} 


The proof of our lower bound uses information theory: one can lower bound the expected number of questions 
	by the amount of information that the questioner gains. 
	There are two such types of information: first, the hidden object reveals $H(\mu)$ information, 
	as in the setting where no lies are allowed. Second, when the object is revealed, 
	the positions of the lies are revealed as well. 
	This reveals additional $H_2(\mu)$ (conditional) information, as we explain below.
	 
	Let $d_x$ denote the number of questions asked for element $x$. Kraft's inequality shows that any good strategy of the questioner satisfies $d_x \gtrapprox \log(1/\mu(x))$. 
	If the answerer chooses a randomized strategy in which the positions of the lies are chosen uniformly from the $\binom{d_x}{k}$ possibilities, 
	these positions reveal $\log \binom{d_x}{k} \approx k \log d_x \gtrapprox k \log \log (1/\mu(x))$ information given $x$. 
	Taking expectation over $x$, the positions of the lies reveal at least $k H_2(\mu)$ information beyond the identity of $x$.

\subsection{Related work}

Most of the literature on error-resilient search procedures has concentrated on the non-distributional setting, in which the goal is to give a worse case guarantee on the number of questions asked, under various error models. The most common error models are as follows:\footnote{This section is heavily based on Pelc's excellent and comprehensive survey~\cite{Pelc}}
\begin{itemize}
\item Fixed number of errors. This is the error model we consider, and it is also the one suggested by Ulam~\cite{Ulam}. This model was first studied by Berlekamp~\cite{Berlekamp}, who used an argument similar to the sphere-packing bound to give a lower bound on the number of questions. Rivest et al.~\cite{Rivest} used this lower bound as a guiding principle in their almost matching upper bound using comparison queries.
\item At most a fixed fraction $p$ of the answers can be lies. This model is similar to the one considered in error-correcting codes. Pelc~\cite{Pelc87} and Spencer and Winkler~\cite{SpencerWinkler} (independently) gave a non-adaptive strategy for revealing the hidden element when $p \leq 1/4$, and showed that the task is not possible (non-adaptively) when $p > 1/4$. Furthermore, when $p < 1/4$ there is an algorithm using $O(\log n)$ questions, and when $p = 1/4$ there is an algorithm using $O(n)$ questions, which are both optimal (up to constant factors). Spencer and Winkler also showed that if questions are allowed to be adaptive, then the hidden element can be revealed if and only if $p < 1/3$, again using $O(\log n)$ questions. 
\item At most a fixed fraction $p$ of any \emph{prefix} of the answers can be lies. Pelc~\cite{Pelc87} showed that the hidden element can be revealed if and only if $p < 1/2$, and gave an $O(\log n)$ strategy when~$p < 1/4$. Aslam and Dhagat~\cite{AslamDhagat} and Spencer and Winkler gave an $O(\log n)$ strategy for all $p < 1/2$.
\item Every question is answered erroneously with probability $p$, an error model common in information theory. R\'enyi~\cite{Renyi} showed that the number of questions required to discover the hidden element with constant success probability is $(1+o(1)) \log n/(1-h(p))$.
\end{itemize}


The distributional version of the ``twenty questions'' game (without lies) was first considered by Shannon~\cite{Shannon} in his seminal paper introducing information theory, where its solution was attributed to Fano (who published it later as~\cite{Fano}). The Shannon--Fano code uses at most $H(\mu)+1$ questions on average, but the questions can be arbitrary. The Shannon--Fano--Elias code (also due to Gilbert and Moore~\cite{GilbertMoore}), which uses only comparison queries, asks at most $H(\mu)+2$ questions on average. Dagan et al.~\cite{dfgm} give a strategy, using only comparison and equality queries, which asks at most~$H(\mu)+1$ questions on average.

\paragraph{Sorting} The non-distributional version of sorting has also been considered in some of the settings considered above:
\begin{itemize}
\item At most $k$ errors: Lakshmanan et al.~\cite{LRG91} gave a lower bound of $\Omega(n\log n + kn)$ on the number of questions, and an almost matching upper bound of $O(n\log n + kn + k^2)$ questions. An optimal algorithm, using $n\log n + O(kn)$ questions, was given independently by Bagchi~\cite{Bagchi} and Long~\cite{Long}. 
\item At most a $p$ fraction of errors in every prefix: Aigner~\cite{Aigner} showed that sorting is possible if and only if $p < 1/2$. Borgstrom and Kosaraju~\cite{BK93} had showed earlier that even \emph{verifying} that an array is sorted requires $p < 1/2$.
\item Every answer is correct with probability $p$: Feige et al.~\cite{FPRU94} showed in an influential paper that $\Theta(n\log(n/\epsilon))$ queries are needed, where $\epsilon$ is the probability of error.
\item Braverman and Mossel~\cite{BM09} considered a different setting, in which an algorithm is given access to noisy answers to all possible $\binom{n}{2}$ comparisons, and the goal is to find the most likely permutation. They gave a polynomial time algorithm which succeeds with high probability.
\end{itemize}

The distributional version of sorting (without lies) was considered by Moran and Yehudayoff~\cite{MY13}, who gave a strategy using at most $H(\mu)+2n$ queries on average, based on the Gilbert--Moore algorithm.

\paragraph{Paper organization.} After a few preliminaries in Section~\ref{sec:definitions}, we describe our results in full in Section~\ref{sec:mr}. We prove our lower bound on the number of questions in Section~\ref{sec:lower-bound}. We present our main upper bound in Section~\ref{sec:ub1}, and an improved version in Section~\ref{sec:2ndalgproof}. We close the paper with a discussion of sorting in Section~\ref{sec:sort}.

\section{Definitions} \label{sec:definitions}

We use the notation $\binom{n}{\leq k} = \sum_{\ell=0}^k \binom{n}{\ell}$.
Unless stated otherwise, all logarithms are base~2. We define $\olog(x) = \log (x + C)$ and $\oln(x) = \ln(x+C)$ for a fixed sufficiently large constant $C >0$ satisfying~$\log \log \log C > 0$.

\paragraph{Information theory.}

Given a probability distribution $\mu$ with countable support, the entropy of~$\mu$ is given by the formula
\[
H(\mu) = \sum_{x \in \supp\mu} \mu(x) \log \frac{1}{\mu(x)}.
\]

\paragraph{Twenty questions game.}
We start with an intuitive definition of the game, played by a questioner (Alice) and an answerer (Bob). Let $U$ be a finite set of elements, and let $\mu$ be a distribution over $U$, known to both parties. The game proceeds as follows: first, an element $x \sim \mu$ is drawn and revealed to the answerer but not to the questioner. The element $x$ is called the \emph{hidden} element. The questioner asks binary queries of the form ``$x \in Q$?'' for subsets $Q \subseteq U$. The answerer is allowed to lie a fixed number of times, and the goal of the questioner is to recover the hidden element $x$, asking the minimal number of questions on expectation.

\paragraph{Decision trees.}

Let $U$ be a finite set of elements. A \emph{decision tree} $T$ for $U$ is a binary tree formalizing the question asking strategy in the twenty questions game. Each internal node of $v$ of~$T$ is labeled by \emph{a query} (or \emph{question}) --- a subset of $U$, denoted by $Q(v)$;
and each leaf is labeled by the output of the decision tree, which is an element of $U$.
The semantics of the tree are as follows: on input $x \in U$, traverse the tree by starting at the root, and whenever at an internal node $v$, go to the left child if $x \in Q(v)$ and to the right child if $x \notin Q(v)$.

\paragraph{Comparison tree.}
Given an ordered set of elements $x_1 \prec x_2 \prec \cdots \prec x_n$, \emph{comparison questions} are questions of the form $Q = \{ x_1, \dots, x_{i-1} \}$, for some $i = 1, \dots, n+1$. In other words, the questions are ``$x \prec x_i$?'' for some $i = 1, \dots, n+1$. An answer to a comparison question is one of $\{ \prec, \succeq \}$. A \emph{comparison tree} is a decision tree all of whose nodes are labeled by comparison questions.

\paragraph{Adversaries.}

Let $k\geq 0$ be a bound on the number of lies. 
An intuitive way to formalize the possibility of lying is via an adversary.
The adversary knows the hidden element $x$ and 
receives the queries from the questioner as the tree is being traversed.
The adversary is allowed to lie at most~$k$ times, where each lie is a violation of the above stated rule. 
Formally, an adversary is a mapping that receives as input an element $x\in X$, 
a sequence of the previous queries and their answers, 
and an additional query $Q\subseteq U$, which represents the current query.
The output of the adversary is a boolean answer to the current query;
this answer is a \emph{lie} if it differs from the truth value of~``$x\in Q$''.

We also allow the adversary and the tree to use randomness: a randomized decision tree is a distribution over decision trees
and a randomized adversary is a distribution over adversaries. 

\paragraph{Computation and complexity.}
The responses of the adversary induce a unique root-to-leaf path in the decision tree, which results in the output of the tree.
A decision tree is \emph{$k$-valid} if it outputs the correct element against any adversary that lies at most $k$ times. 

Given a $k$-valid decision tree $T$ and a distribution $\mu$ on $U$, the \emph{cost} of $T$ with respect to $\mu$, denoted $c(T,\mu)$, is the maximum, over all possible adversaries that lie at most $k$ times, of the expected\footnote{The expectation is also taken with respect to the randomness of the adversary and the tree when they are randomized.}
length of the induced root-to-leaf path in $T$.
Finally, the $k$-cost of $\mu$, denoted $c_k(\mu)$, is the minimum of $c(T,\mu)$ over all $k$-valid decision trees~$T$.

\paragraph{Basic facts.}
We will refer to the following well-known formula as \emph{Kraft's identity}:

\begin{fact}[Kraft's identity] \label{cla:sum-leaves}
	Fix a binary tree $T$, let $L$ be its set of leaves and let $d(\ell)$ be the depth of leaf $\ell$. The following applies:
	\[\sum_{\ell \in L(T)} 2^{-d(\ell)}\leq 1.\]
\end{fact}

We will use the following basic lower bound on the expected depth by the entropy:
\begin{fact}\label{fact:ent}
	Let $T$ be a binary tree and let $\mu$ be a distribution over its leaves. 
	Then 
	\[H(\mu)\leq \EE_{\ell\sim \mu}\bigl[d(\ell)\bigr].\] 
\end{fact}
In other words, for any distribution $\mu$, $c_0(\mu) \ge H(\mu)$. In fact, it is also known that $c_0(\mu) \le H(\mu) + 1$.

\section{Main results} \label{sec:mr}

This section is organized as follows:
The lower bound is presented in Section~\ref{sec:mr:lb}. 
	Then, the two searching algorithms are presented in Section~\ref{sec:mr:ub}, 
	and finally the application to sorting is presented in Section~\ref{sec:mr:sort}.

\subsection{Lower bound} \label{sec:mr:lb}

In this section we present the following lower bound on $c_k(\mu)$, namely, on the expected number of questions asked by \emph{any} $k$-valid tree (not necessarily a comparison trees).

\begin{theorem} \label{thm:lower-bound-inf}
	For every non-constant distribution $\mu$ and every $k \geq 0$,
	\[
	c_k(\mu) \geq \Bigl(\EE_{x\sim \mu}\log \frac{1}{\mu(x)}\Bigr) + k \Bigl(\EE_{x \sim \mu} \log \log \frac{1}{\mu(x)}\Bigr)  -  (k\log k + k + 1).
	\]
\end{theorem}

The proof of this lower bound appears in Section~\ref{sec:lower-bound}.

\paragraph{Proof overview.}
Consider a $k$-valid tree; we wish to lower bound the expected number of questions for $x\sim \mu$.
Let $d_x$ denote the number of questions asked when the secret element is $x$.
Then, by the entropy lower bound when the number of mistakes is $k=0$,
it follows that {\it typically},~$d_x \gtrsim\log(1/\mu(x))$.
Moreover, the transcript of the game (i.e.\ the list  of questions and answers)
determines both~$x$ and the positions of the $k$ lies.
This requires 
\[d_x + k\log (d_x) \gtrsim \log(1/\mu(x)) + k\log\log(1/\mu(x))\] 
bits of information.
Taking expectation over $x\sim\mu$ then yields the stated bound. 

Our proof formalizes this intuition using standard and basic tools from information theory.
One part that requires a subtler argument is
showing that indeed one may assume that $d_x\gtrsim \log(1/\mu(x))$ for all $x$.
This is done by showing that any $k$-valid tree can be modified to satisfy this constraint
without increasing the expected number of questions by too much.
The crux of this argument, which relies on Kraft's identity (Fact~\ref{cla:sum-leaves}), appears in Lemma~\ref{lem:alpha}.

\subsection{Upper bounds} \label{sec:mr:ub}

We introduce two algorithms.
The first algorithm, presented in Section~\ref{sec:mr:ub1}, is simpler, however, the second algorithm has a better query complexity. The expected number of questions asked by the first algorithm is at most
\[
H(\mu) + (k+1) H_2(\mu) + O(k^2 \Hlogloglog(\mu) + k^2 \log k),  \quad \text{ where } \Hlogloglog(\mu) = \sum_x \mu(x) \log \log \log \frac{1}{\mu(x)}.
\]
The second algorithm, presented in Section~\ref{sec:mr:ub2}, removes the quadratic dependence on $k$, and has an expected complexity of:
\[
H(\mu) + k \Hloglog(\mu) + O(k\Hlogloglog(\mu) + k \log k).
\]
In Section~\ref{sec:mr:finegrained} we robustify the guarantees of these algorithms
and consider scenarios where the exact distribution $\mu$ is not known but only some prior $\eta\approx\mu$,
or where the actual number of lies is less than the bound $k$ 
(whence the algorithm achieves better performance).

\subsubsection{First algorithm} \label{sec:mr:ub1}

Suppose that we are given a probability distribution $\mu$ whose support is the linearly ordered set $x_1 \prec \cdots \prec x_n$. 
	In this section we overview the proof of the following theorem (the complete proof appears in Section~\ref{sec:ub1}):

\begin{theorem} \label{thm:ub-main}
	There is a $k$-valid comparison tree $T$ with
	\[
	c(T,\mu) \le H(\mu) + (k+1)\sum_{i=1}^n \mu_i \log \olog \frac{1}{\mu_i} + O\left(k^2 \sum_{i=1}^n \mu_i \log \log \olog \frac{1}{\mu_i} + k^2 \log k \right),
	\]
	where $\mu_i = \mu(x_i)$.
\end{theorem}

The question-asking strategy simulates a binary search to recover the hidden element. 
	If, at some point, the answer to some question $q$ is suspected as a lie then $q$ is asked $2k+1$ times to verify its answer. 
	When is the answer to $q$ suspected? 
	The binary search tree is constructed in a manner that if no lies are told 
	then roughly half of the questions are answered $\prec$, and half $\succeq$. 
	However, if, for example, the lie ``$x \succeq x_{50}$'' is told when in fact $x = x_{10}$, 
	then all consecutive questions will be of the form ``$x \prec x_i$?'' for $i > 50$, 
	and the correct answer would always be $\prec$. 
	Since no more than $k$ lies can be told, almost all consecutive questions will be answered $\prec$, 
	and the algorithm will suspect that some earlier question is a lie.

We start by suggesting a question-asking strategy using comparison queries which is valid as long as there are no lies, and then show how to make it resilient to lies.
Each element $x_i$ is mapped to a point $p_i$ in $[0,1]$, such that $p_1 < p_2 < \cdots < p_n$. Then, a binary search on the interval $[0,1]$ is performed, for finding the point $p_i$ corresponding to the hidden element.
The search proceeds by maintaining a $\Live$ interval, which is initialized to $[0,1]$. At any iteration, the questioner asks whether $p_i$ lies in the left half of the $\Live$ interval. The interval is updated accordingly, and its length shrinks by a factor of $2$. This technique was proposed by Gilbert--Moore~\cite{GilbertMoore}, and is presented in \auxAlgRef~\ref{alg:gilbert-moore}, as an algorithm which keeps asking questions indefinitely.

\begin{auxAlg}
	\begin{algorithmic}[1]
		\State{$\Live \gets [0,1]$}
		\Loop
		\State{$m \gets \text{midpoint of } \Live$}
		\State{$X \gets \{ i : p_i \geq m \}$}
		\If{$x \in X$}
		\State{$\Live \gets \text{right half of } \Live$}
		\Else
		\State{$\Live \gets \text{left half of } \Live$}
		\EndIf
		\EndLoop
	\end{algorithmic}
	\caption{Randomized Gilbert--Moore} \label{alg:gilbert-moore}
\end{auxAlg}

The points $p_1, \dots, p_n$ are defined as follows: first, a number $\theta \in [0,1/2)$ is drawn uniformly at random. Now, for any element $i$ define $p_i = \frac{1}{2} \sum_{j=1}^{i-1} \mu_j + \frac{1}{4}\mu_i + \theta$.\footnote{In the original paper $p_i$ was defined similarly but without the randomization: $p_i = \sum_{j=1}^{i-1} \mu_j + \frac{1}{2} \mu_i$.} Given $\theta$, let $T'_\theta$  denote the infinite tree generated by \auxAlgRef~\ref{alg:gilbert-moore}. Note that whenever $\Live$ contains just one point $p_i$, 
then (as there are no lies) the hidden element must be $x_i$. Denote by $T_\theta$ the finite tree corresponding to the algorithm which stops whenever that happens.
We present two claims about these trees which are proved in Section~\ref{sec:ub:stepone}.

First, conditioned on any hidden element $x_i$, the answers to all questions (except, perhaps, for the first answer) are distributed uniformly and independently, where the distribution is over the random choice of $\theta$. This follows from the fact that all bits of $p_i$ except for the most significant bit are i.i.d.\ unbiased coin flips.

\begin{claim} \label{cla:gilbert-moore-uniform}
	For any element $x_i$, let $(A_t)$ be the random sequence of answers to the questions in \auxAlgRef~\ref{alg:gilbert-moore}, containing all answers except for the first answer, assuming there are no lies. The distribution of the sequence $(A_t)$ is the same as that of an infinite sequence of independent unbiased coin tosses, where the randomness stems from the random choice of $p_i$.
\end{claim}

Second, since $\min(p_i - p_{i-1}, p_{i+1} - p_i) \ge \mu_i/4$, one can bound the time it takes to isolate $x_i$ as follows.

\begin{claim} \label{cla:gilbert-moore-depth}
	For any element $x_i$ and any $\theta$, the leaf in $T_\theta$ labeled by $x_i$ is of depth at most $\log (1/\mu_i) + 3$. Hence, if $x$ is drawn from a distribution $\mu$, the expected depth of the leaf labeled $x$ is at most $\sum_i \mu_i \log (1/\mu_i) + 3 = H(\mu)+3$.
\end{claim}

\setcounter{algorithm}{0}
\begin{algorithm}
	\begin{algorithmic}[1]
		\State{$\theta \gets \mathrm{Uniform}([0,1/2))$}
		\State{$\Current \gets \theroot(T'_\theta)$}
		\State{$\LastVerified \gets \theroot(T_\theta')$}
		\While{$\LastVerified$ is not a leaf of $T_\theta$}
		\If{$x \in Q(\Current)$}
		\State{$\Current \gets \leftchild(\Current)$}
		\Else
		\State{$\Current \gets \rightchild(\Current)$}
		\EndIf
		\State{$d \gets \depth(\LastVerified) + 1$}
		\If{$\depth(\Current) = d + r(d)$}
		\State $\Candidate \gets $ child of $\LastVerified$ which is an ancestor of $\Current$
		\State $\VerificationPath \gets $ ancestors of $\Current$ up to and excluding $\Candidate$
		\If{$\Candidate$ is a left (right) child and at most $k-1$ vertices in $\VerificationPath$ are left (right) children}
		\State{Ask $2k+1$ times the question $x \in Q(\LastVerified)$} \label{l:majority-vote} 
		\If{majority answer is $x \in Q(\LastVerified)$}
		\State{$\LastVerified \gets \leftchild(\LastVerified)$}
		\Else
		\State{$\LastVerified \gets \rightchild(\LastVerified)$}
		\EndIf
		\State{$\Current \gets \LastVerified$}
		\Else
		\State{$\LastVerified \gets \Candidate$}
		\EndIf
		\EndIf
		\EndWhile
		\State \Return label of $\LastVerified$
	\end{algorithmic}
	\caption{Resilient-Tree} \label{alg:main-ub}
\end{algorithm}

We now describe the $k$-resilient algorithm: Algorithm~\ref{alg:main-ub} (the pseudocode appears as well). At the beginning, a number $\theta$ is randomly drawn. Then, two concurrent simulations over $T'_\theta$ are performed, and two pointers to nodes in this tree are maintained (recall that $T'_\theta$ is the infinite binary search tree).
The first pointer, $\Current$, simulates the question-asking strategy according to $T'_\theta$, ignoring the possibility of lies. In particular, it may point on an incorrect node in the tree (reached as a result of a lie).
Since $\Current$ ignores the possibility of lies, there is a different pointer, $\LastVerified$, which verifies the answers to the questions asked in the simulation of $\Current$. All answers in the path from the root to $\LastVerified$ are verified as correct, and $\LastVerified$ will always be an ancestor of $\Current$. See Figure~\ref{fig:algorithm-1} for the basic setup.

\begin{figure}
\centering
\includegraphics{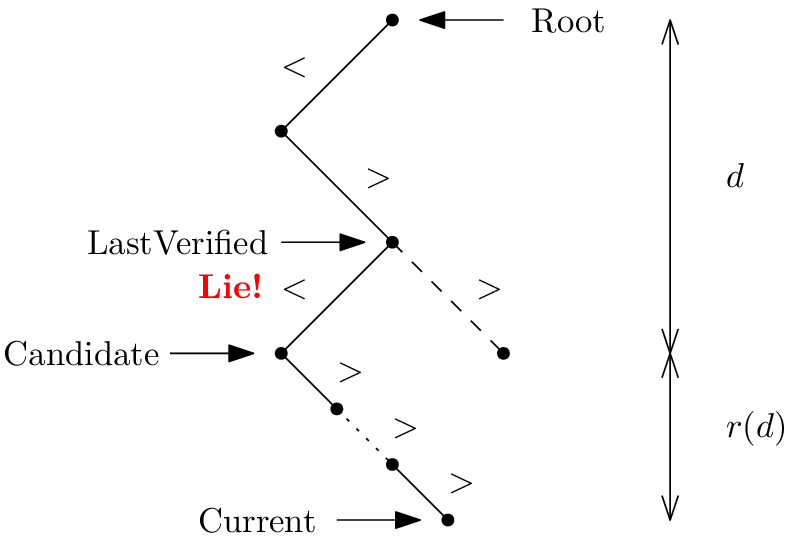}
\caption{An illustration of Algorithm~\ref{alg:main-ub} just before the detection of a lie. The answer at $\LastVerified$ was a lie ($<$ instead of $>$), and so all answers below $\Candidate$ (except for any further lies) are $>$. This is noticed since $\Current$ is at depth $d + r(d)$. The answer at $\Candidate$ will be verified and found wrong, and so $\LastVerified$ would move to the sibling of $\Candidate$ (and so will $\Current$), and the algorithm will continue from that point.}
\label{fig:algorithm-1}	
\end{figure}

The algorithm proceeds in iterations. 
	In every iteration the question $Q(\Current)$ is asked and $\Current$ is advanced to the corresponding child accordingly. 
	In \underline{some} of the iterations also $\LastVerified$ is advanced. 
	Concretely, this happens when the depth of $\Current$ in~$T'_\theta$ equals $d + r(d)$, 
	where $d$ is the depth of $\LastVerified$ and $r(d) \approx \log d + k \log \log d$.\footnote{The exact definition of $r(d)$ is in Equation~\eqref{eq:def-rd}.} 
	In these iterations, the answer given to $Q(\LastVerified)$ is being verified, as detailed next. 

\paragraph{The verification process.}
Next, we examine the verification process when $\LastVerified$ is advanced. 
	There are two possibilities: 
	first, when the answer to the question $Q(\LastVerified)$, 
	which was given when $\Current$ traversed it, is verified to be correct.
	In that case, $\LastVerified$ moves to its child which lies on the path towards $\Current$. 
	In the complementing case, when the the answer to the question $Q(\LastVerified)$ is detected as a lie, 
	then $\LastVerified$ moves to the other child. In that case, $\Current$ is no longer a descendant of $\LastVerified$, 
	hence $\Current$ is moved up the tree and is set to $\LastVerified$.

We now explain how the answer to $Q(\LastVerified)$ is verified. 
	 There are two verification steps: the first step uses no additional questions 
	 and the second step uses $2k+1$ additional questions.
	 Usually, only the first step will be used and no additional questions will be spent during verification.
	 In the first verification step one checks whether the following condition holds: 
	\begin{center}
	 {\it The answer to $Q(\LastVerified)$ is identical to at least $k$ of the answers along the path from $\LastVerified$ to $\Current$}.
	 \end{center}
	 If this condition holds, then the answer is verified as correct. 
	To see why this reasoning is valid, assume without loss of generality that the answer is $\prec$, 
	 and assume towards contradiction that it was a lie. 
	 Then, the correct answers to all following questions in the simulation of $\Current$ are~$\succeq$. 
	 Since there can be at most $k-1$ additional lies, there can be at most $k-1$ additional~$\prec$ answers. Hence, 
	 if there are more $\prec$ answers among the following questions then the previous answer to~$Q(\LastVerified)$ is verified as correct. 
	 
Else, if the above condition does not hold then one proceeds to the second verification step and asks~$2k+1$ times
	the question~$Q(\LastVerified)$. Here, the majority answer must be correct, 
	since there can be at most $k$ lies. 
	
We add one comment: if the second verification step is taken, 
	one sets $\Current \gets \LastVerified$ regardless of whether a lie had been revealed (this is performed to facilitate the proof).
	So, whenever the condition in the first verification step fails to hold then $\Current$ and $\LastVerified$ point to the same node 
	in the tree.

The algorithm ends when $\LastVerified$ reaches a leaf of $T_\theta$, at which point the hidden element is recovered.

\paragraph{Query complexity analysis.}
Fix an element $x_i$. 
	We bound the expected number of questions asked when $x_i$ is the hidden element as follows. 
	Define $d \approx \log(1/\mu(x_i))$ as the depth of the leaf labeled~$x_i$ in~$T_\theta$. 
	We divide the questions into the following five categories:
\begin{itemize}
	\item Questions on the path $P$ from the root to $\Current$ by the end of the algorithm,
	when $\Current$ reaches depth $d + r(d)$, $\LastVerified$ reaches depth $d$, and the algorithm terminates. Hence, there are at most $d + r(d)$ such questions.
	\item Questions that were ignored due to the second verification step
	while $\Current$ was backtracked from a node outside $P$. 
	This can only happen due to a lie between $\Current$ and $\LastVerified$ so there are at most $k\cdot r(d)$ such questions.
	\item Questions asked $2k+1$ times during the second verification step when $\Current$
	was pointing to a node outside $P$.
	This can only happen due to a lie between $\Current$ and $\LastVerified$ so there are at most $k\cdot(2k+1)$ such questions.
	\item Questions that were ignored due to the second verification step,
	when $\Current$ was being backtracked from a node in $P$.
	By the choice of $r(d)$ there are at most $O(1)$ such questions (on expectation).
	\item Questions asked $2k+1$ during the second verification step when $\Current$
	was pointing to a node in $P$.
	By the choice of $r(d)$ there are at most $O(1)$ such questions (in expectation).
\end{itemize}
Summing these bounds up, one obtains a bound of 
\[ \bigl(d+r(d)\bigr) + k\cdot r(d) + k\cdot (2k+1) + O(1) + O(1) \approx \log (1/\mu_i) + (k+1) \Bigl(\log \log (1/\mu_i) + k \log \log \log (1/\mu_i) + O(k)\Bigr). \]

\subsubsection{Second algorithm} \label{sec:mr:ub2}

In this section we overview the proof of the following theorem (the complete proof appears in Section~\ref{sec:2ndalgproof}).

\begin{theorem} \label{thm:algtwo-better}
	For any distribution $\mu$ there exists a $k$-valid comparison tree $T$ with
	\[
	c(T,\mu) \le H(\mu) + k \mathbb{E}_{x \sim \mu}\left[\log \olog (1/\mu(x))\right] + O( k\mathbb{E}_{x \sim \mu}\left[\log \log \olog (1/\mu(x)) \right] + k \log k).
	\]
\end{theorem}

We explain the key differences with \algone{}. 

\begin{itemize}
	\item In Algorithm~\ref{alg:main-ub}, 
	an answer to a question $Q$ at depth $d$ was suspected as a lie if at most $k$ of the $r(d)$ consecutive questions received the same answer as $Q$. 
	In the new algorithm, we suspect a question $Q$ if \emph{all} the $r'(d)$ consecutive answers are different than $Q$.
	This change enables setting $r'(d) \approx \log d$ rather than the previous value of $r(d) \approx \log d + k \log \log d$. Similarly to \algone{}, any time a lie is deleted, $r'(d)$ questions are being deleted. Summing over the $k$ lies, one obtains a total of $k r'(d) \approx k \log d$ deleted questions, which is smaller than the corresponding value of $k r(d) \approx k \log d + k^2 \log \log d$ in \algone{}.
	
	\item In \algone{}, the lies were detected in the same order they were told (i.e.\ in a {\it first-in-first-out} queue-like manner).
	This is due to the semantic of the pointer $\LastVerified$ which verifies the questions one-by-one, along the branching of the tree. 
	In \algtwo{} the pointer $\LastVerified$ is removed (only $\Current$ is used), and the lies are detected in a
	{\it last-in-first-out} stack-like manner: only the last lie can be deleted at any point in time. Indeed, as described in the previous paragraph, a lie will be deleted only if all consecutive answers are different (which is equivalent to them being non-lies).
	
	\item In \algone{}, when an answer is suspected as a lie, 
	the corresponding question $Q$ is repeated~$2k+1$ times in order to verify its correctness. 
	This happens after each lie, hence $\Omega(k)$ redundant questions are asked per lie. 
	In \algtwo{}, the suspected question $Q$ will be asked again only \emph{once}, 
	and the algorithm will proceed accordingly. 
	It may however be the case that this process will repeat itself and also the second answer to this question will be suspected as a lie
	and  $Q$ will be asked once again and so on.
	In order to avoid an infinite loop we add the condition that
	if the same answer is told $k+1$ times then it is guaranteed to be correct 
	and will not be suspected any more.
	
	\item The removal of $\LastVerified$ forces finding a different method of verifying the correctness of an element $x$ upon arriving at a leaf of $T_\theta$. One option is to ask the question ``element $=x$?'' $2k+1$ times and take the majority vote, where each $=$ question is implemented using one $\preceq$ and one $\succeq$. This will, however, lead to asking $\Omega(k)$ redundant questions each time $x$ is not the correct element. Instead, one asks ``element $=x$?'' multiple times, stopping either when the answer $=$ is obtained $k+1$ times, or by the first the answer $\ne$ has obtained more than the answer $=$. The total redundancy imposed by these verification questions throughout the whole search is $O(k)$.
\end{itemize}
To put the algorithm together, we exploit some simple combinatorial properties of paths containing multiple lies.

\subsubsection{A fine-grained analysis of the guarantees} \label{sec:mr:finegrained}

In this section, we present a stronger statement for the guarantees of our algorithms. First, the algorithms do not have to know \emph{exactly} the distribution $\mu$ from which the hidden element is drawn: an approximation suffices for getting a similar bound. Recall that the algorithm gets as an input some probability distribution $\eta$. This distribution might differ from the true distribution $\mu$. The cost of using $\eta$ rather than $\mu$ is related to $D(\mu \| \eta)$, the Kullback--Leibler divergence between the distributions.

Secondly, the algorithm has stronger guarantees when the actual number of lies is less than~$k$. This is an improvement comparing to the algorithm of Rivest et al.~\cite{Rivest} mentioned in the introduction. It will be utilized in the application of sorting, where the searching algorithm is invoked multiple times with a bound on the total number of lies (rather the number of lies per iteration). We present the general statement with respect to \algtwo{}. The statement and the corresponding bound on \algone{} appears in Lemma~\ref{lem:ub-main} in Section~\ref{sec:ub1}.

\begin{theorem} \label{thm:algtwo-one-elem}
	Assume that \algtwo{} is invoked with the distribution $(\eta_1, \dots, \eta_n)$. Then, for any element $x_i$, the expected number of questions asked when $x_i$ is the secret is at most
	\[ \log\left(1/\eta_i \right) + \EE[K']  \log \olog \frac{1}{\eta_i } + O\left(\EE[K'] \log \log \olog \frac{1}{\eta_i} + \EE[K'] \olog k + k \right), \]
	where $K'$ is the expected number of lies. (The expectation is taken over the randomness of both parties.) 
\end{theorem}

As a corollary, one obtains Theorem~\ref{thm:algtwo-better} and the following corollary, which corresponds to using a distribution different from the actual distribution.

\begin{corollary} \label{cor:diffdist}
	Assume that \algtwo{} is invoked with $(\eta_1, \dots, \eta_n)$ while $(\mu_1, \dots, \mu_n)$ is the true distribution. Then, for a random hidden element drawn from $\mu$, the expected number of questions asked is at most
	\begin{align*}
	H(\mu) &+ k \mathbb{E}_{x \sim \mu}\left[\log \olog (1/\mu(x))\right] + O( k \EE_{\mu}\left[\log \log \olog (1/\mu(x)) \right] + k \olog k ) \\
	&+ D(\mu \| \eta) + O(k \log D(\mu \| \eta)),
	 \end{align*}
	where $D(\mu \| \eta) = \sum_{x \in \supp{\mu}} \mu(x) \log \frac{\mu(x)}{\eta(x)}$ is the Kullback--Leibler divergence between $\mu$ and $\eta$.
\end{corollary}

Corollary~\ref{cor:diffdist} follows from Theorem~\ref{thm:algtwo-one-elem} by bounding $K' \le k$, taking expectation over $x_i \sim \mu$, noting that $\sum_i \mu_i \log (1/\eta_i) = H(\mu) + D(\mu \| \eta)$ and applying Jensen's inequality with the function $x \mapsto \log x$.

\subsection{Sorting} \label{sec:mr:sort}

One can apply \algtwo{} to implement a stable version of the insertion sort using comparison queries. Let $\Pi$ be a distribution over the set of permutations on $n$ elements. Complementing with prior algorithms achieving a complexity of $H(\Pi) + O(n)$ in the randomized setting with no lies \cite{MY13}, and $n \log n + O(nk+n)$ in the deterministic setting with $k$ lies \cite{Bagchi,Long}, we present an algorithm with a complexity of $H(\Pi) + O(nk+n + k \log k)$ in the distributed setting with $k$ lies. Note that $k \log k = O(nk)$ unless unless the unlikely case that $k = e^{w(n)}$, hence the $k \log k$ term can be ignored. Therefore, the guarantee of our algorithm matches the guarantees of the prior algorithms substituting either $k=0$ or $\Pi = \mathrm{Uniform}$.

\begin{theorem} \label{thm:sort-ent}
	Assume a distribution $\Pi$ over the set of all permutations on $n$ elements. There exists a sorting algorithm which is resistant to $k$ lies and sorts the elements using $H(\Pi) + O(nk + n + k \log k)$ comparisons on expectation.
\end{theorem}


(The proof appears in Section~\ref{sec:sort}.) The randomized algorithms benefit from prior knowledge, namely, when one has information about the correct ordering. This is especially useful for maintaining a sorted list of elements, a procedure common in many sequential algorithms. In these settings, the values of the elements can change in time, hence, the elements have to be re-sorted regularly, however, their locations are not expected to change drastically.

The suggested sorting algorithm performs $n$ iterations of insertion sort. By the end of each iteration $i$, $x_1, \dots, x_i$ are successfully sorted. Then, on iteration $i+1$, one performs a binary search to find the location where $x_{i+1}$ should be inserted, using conditional probabilities.

The guarantee of the algorithm is asymptotically tight: a lower bound of $H(\Pi)$ follows from information theoretic reasons, and a lower bound of $\Omega(nk)$ follows as well: the bound of Lakshmanan et al.~\cite{LRG91} can be adjusted to the randomized setting.

\section{Lower bound} \label{sec:lower-bound}

In this section, we prove Theorem~\ref{thm:lower-bound-inf}.

Let $\mu$ be a distribution over $U$ and let $\cal T$ be a (possibly randomized) $k$-valid decision tree with respect to $\mu$.
First we assume that with probability $1$, 
the number of questions asked on any $x\in \supp \mu$ is at least $\alpha(x)=\lceil \log(1/\mu(x))/2 \rceil$. 
This assumption will later be removed. 

Consider a randomized adversary that, after seeing the secret element $x\sim \mu$, 
picks uniformly at random a subset of at most~$k$ questions to lie on from the first $\alpha(x)$ questions.

Let $Q$ denote the random variable of the transcript of the game (i.e.\ the sequence of queries and the answers provided by the adversary), 
and leet $\lvert Q \rvert$ denote its length (i.e.\ the number of query/answer pairs).
Let $X$ denote the random variable of the secret element, 
and let $L$ denote the random variable describing the positions of the lies 
(so, $L$ is a subset of size at most~$k$ of $\{1,\ldots,\alpha(x)\}$). 
Note that $Q$ determines both $X$ and $L$ (since $\cal T$ is $k$-valid).
Therefore,
\[				\EE[\lvert Q \rvert] \geq H(Q) \geq H(X,L) = H(X) + H(L | X) = \EE_{x\sim \mu}\log\frac{1}{\mu(x)} +  H(L | X),\]
where the first inequality is due to Fact~\ref{fact:ent}. 
Now, 
\begin{align*}	
H(L | X) &= \sum_x \mu(x) \log \binom{\alpha (x)} {\leq k} \\
	      &\geq \sum_x \mu(x) \log\bigl( (\alpha(x) / k)^k\bigr)\\
	      &\geq k \EE_{x \in \mu} \log \log \frac{1}{\mu(x)}  -  (k\log k + k),  
\end{align*} 
where the first inequality is due to the well-known formula $\binom{n}{\leq m} \geq (n/m)^m$.
This finishes the proof under the assumption that the number of questions is at least $\alpha(x)$ for every $x\in\supp \mu$.

We next show that this assumption can be removed:
we will show that any $k$-valid tree $T$, 
can be transformed to a $k$-valid tree $T'$  that satisfies this assumption 
and $c(T',\mu)\leq c(T,\mu) + 1$.
It suffices to show this for a deterministic $k$-valid tree $T$, 
since a randomized $k$-valid tree is a distribution over deterministic $k$-valid trees.
%
The lower bound on $c({\cal T}, \mu)$ then follows
from the lower bound on $\cal T'$ plus the additive factor of $1$
due to the transformation of $\cal T$ to $\cal T'$.

Fix a deterministic $k$-valid tree $T$ and let $V$ be the set of elements for which 
there exists an $x$-labeled leaf with depth less than $\alpha(x)$. 
We will show that $\sum_{x\in V} \mu(x) \alpha(x)\leq 1$. 
This implies that increasing the number of questions asked on $x$ to be at least $\alpha(x)$ for all $x\in V$ 
increases the expected number of questions by at most $\sum_{x\in V} \mu(x)\alpha(x)\leq 1$.

\begin{lemma}\label{lem:alpha}
\[ \sum_{x \in V} \mu(x) \alpha(x) \le 1. \]
\end{lemma}
\begin{proof}
For any $x\in V$, let $d_x < \alpha(x)$ denote the minimum depth of an $x$-labeled leaf in $T$.
By Fact~\ref{cla:sum-leaves}:
\[ 1 \geq \sum_{x\in V} 2^{-d_x} \geq \sum_{x\in V} 2^{-\log(1/\mu(x))/2} = \sum_{x\in V} \sqrt{\mu(x)}.\]
Therefore,
\begin{align*}
\sum_{x\in V} \mu(x) \alpha(x) &\leq \sum_{x\in V} \mu(x) \log(1/\mu(x))/2\\
					      &= \sum_{x\in V} \sqrt{\mu(x)} \cdot\Bigl(\sqrt{\mu(x)} \cdot\log(1/\sqrt{\mu(x)})\Bigr)\\
					      &\leq \sum_{x\in V} \sqrt{\mu(x)} \tag{$\sqrt{t} \log\bigl(1/\sqrt{t}\bigr)\leq H(\mathrm{Ber}(\sqrt{t})) \leq 1$ for all $t\in [0,1]$}\\
					      &\leq 1. \tag{$\sum_{x\in V} \sqrt{\mu(x)} \leq 1$}
\end{align*}
\end{proof}
\paragraph{Remark.}
We could slightly improve the lower bound in Theorem~\ref{thm:lower-bound-inf} to
\[c_k(\mu) \geq H(\mu) + k \EE_{x \in \mu} \log \log \frac{1}{\mu(x)}  -  \bigl(k\log k + \Theta(\sqrt{k}) \bigr),\]
by setting $\eps = 1/\sqrt{k}$, $\alpha(x) = \lceil (1-\eps)\log(1/\mu(x)) \rceil$,
and following a similar argument.

\section{First algorithm} \label{sec:ub1}

In this section, we prove Theorem~\ref{thm:ub-main}. It follows immediately from the following lemma:

\begin{lemma} \label{lem:ub-main}
	Fix a set of elements $x_1 \prec x_2 \prec \cdots \prec x_n$, and fix a probability distribution vector $(\mu_1, \dots, \mu_n)$. There is a $k$-valid comparison tree that for any element $x_i$ asks on expectation at most
	\[
	\log (1/\mu_i) + \EE[K'+1] \log \olog \frac{1}{\mu_i} + \EE[K'+1] O\left(k \log \log \olog \frac{1}{\mu_i} + k \olog k \right),
	\]
	questions, where $K'$ is the number of lies told, and the expectation is over the randomness of both the questioner and the answerer, conditioned on $x_i$ being the hidden element.
\end{lemma}

We rely on the definition of the algorithm from Section~\ref{sec:mr:ub1}.
Since we are proving a bound on a specific algorithm, one can assume that the actions of the adversary are deterministic given the hidden element and the history of questions and answers. The proof of the theorem is in three steps. In the first step, we analyze the randomized nonresistant decision trees $T_\theta$ and $T'_\theta$ defined in Section~\ref{sec:mr:ub1}. In the second step, we analyze the $k$-valid tree. In the final step, we make calculations which bound the expected number of asked questions and conclude the proof.



\subsection{Step 1: Analyzing the nonresistant decision tree} \label{sec:ub:stepone}

In this section, we prove the two claims from Section~\ref{sec:mr:ub1}.
\begin{proof}[Proof of Claim~\ref{cla:gilbert-moore-uniform}]
	Note that $p_i$ is uniform in $\left[\frac{1}{2} \sum_{j=1}^{i-1} \mu_j+ \frac{1}{4}\mu_i, \frac{1}{2} \sum_{j=1}^{i-1} \mu_j + \frac{1}{4}\mu_i + 1/2\right)$. Define $p \mapsto p \bmod 1/2 \colon [0,1) \to [0,1/2)$ in the obvious way:
	\[
	p \bmod 1/2 = \begin{cases}
	p & \text{if } 0 \le p < 1/2, \\
	p-1/2 & \text{if } 1/2 \le p < 1.
	\end{cases}
	\]
	Note that $p_i \bmod 1/2$ is distributed uniformly in $[0,1/2)$ and that the binary representation of $p_i$ equals the binary representation of $p_i \bmod 1/2$, except, perhaps, for the bit which corresponds to $2^{-1}$ (the bit $b_1$ in the binary representation $p_i = 0.b_1b_2\cdots$). Hence, the bits of $p_i$ (except for the first bit) are distributed as an infinite sequence of independent unbiased coin tosses. 
	
	Note that the answer to question no.~$t$ in \auxAlgRef~\ref{alg:gilbert-moore} equals bit $t$ of the binary representation of $p_i$. In particular, the answers (except for the first one) are distributed as independent unbiased coin tosses.
\end{proof}

\begin{proof}[Proof of Claim~\ref{cla:gilbert-moore-depth}]
	Let $d_i = \min\left(p_i - p_{i-1}, p_{i+1} - p_i \right) \ge \mu_i/4$ be the minimal distance of $p_i$ from its neighboring points (assuming $p_0 = 0$ and $p_{n+1} = 1$, for completeness). After $t$ steps of the algorithm, $\Live$	is an interval of width $2^{-t}$ containing $x_i$. Therefore if $2^{-t} \leq d_i$ then $x_i$ is the only point contained in $\Live$. This shows that the depth of the leaf labeled $x_i$ is at most $\lceil \log (1/d_i) \rceil \le \lceil \log (4/\mu_i) \rceil \le \log(1/\mu_i) + 3$.
\end{proof}

\subsection{Step 2: Making a tree resilient}

We give the definition of $r(d)$:
\begin{equation} \label{eq:def-rd}
r(d) = \left\lceil\log \left((d+1)\ln^2(d+1)\right) + 4 (k+1) \left( \log \log ((d+1)\ln^2(d+1)) + 4 \log \frac{k+1}{\ln 2} \right) \right\rceil.
\end{equation}
Note that
\begin{equation} \label{eq:rd-simplified}
r(d) = \log (d+1) + O(k \log \log (d+1) + k \log k + 1).
\end{equation}

As explained in Section~\ref{sec:mr:ub1}, the following holds:

\begin{claim} \label{cla:alg-correct}
	The decision tree corresponding to Algorithm~\ref{alg:main-ub} is $k$-valid.
\end{claim}

To sketch a proof, it suffices to show that $\LastVerified$ always resides in the path $P$ from the root to the leaf labeled by the hidden element $x_i$. This is done by induction on the depth of $\LastVerified$: at the beginning, $\LastVerified$ resides on the root. For the induction step, assume in some iteration that $\LastVerified$ moves from a node $v$ of depth $d-1$ to its child $v'$ of depth $d$, and assume without loss of generality that $v'$ is a left child. The pointer $\LastVerified$ could only move to $v'$ after it is verified that $x_i \in Q(v)$: either after at least $k$ matching answers were obtained after $Q(v)$ was asked, or after $Q(v)$ was asked $2k+1$ times. Hence, $v'$ is in $P$ as required and the proof follows.

We proceed to determine the expected number of iterations it takes the algorithm to determine an element $x_i$. Let $P_\theta$ be the unique root-to-leaf path in $T_\theta$ which is consistent with $x_i$, and let $D_\theta$ denote its depth. We define the notion of a problematic node: a node $v$ on $P_\theta$ which may be suspected by the algorithm as not residing in $P_\theta$. Formally, a node $v \in P_\theta$ which is a left (right) child is \emph{problematic for $x_i$} if among its first $r(d)$ descendants there at most $k$ left (right) children. 
Let $F_\theta$ be the number of vertices on $P_\theta$ which are problematic for $x_i$, and let $F'_\theta = \sum_v r(\depth(v))$, where the sum goes over these vertices. 

\begin{lemma} \label{lem:alg-running-time}
	Let $x_i$ be an element, and fix some $0 \le \theta \le 1/2$. Algorithm~\ref{alg:main-ub} terminates after at most this many steps:
	\[
	D_\theta + r(D_\theta) + F_\theta(2k+1) + F'_\theta + K'(r(D_\theta)+2k+1).
	\]	
\end{lemma}

\begin{proof}
	We divide the questions into categories, and bound each separately:
	\begin{itemize}
		\item Questions on the path $P_\theta$ from the root to $\Current$ by the end of the algorithm:
		when $\Current$ reaches depth $D_\theta + r(D_\theta)$, $\LastVerified$ reaches depth $D_\theta$ and the algorithm terminates. Hence, there are at most $D_\theta + r(D_\theta)$ such questions.
		\item Questions that were ignored due to the second verification step
		while $\Current$ was backtracked from a node outside $P_\theta$. 
		This can only happen due to a lie between $\Current$ and $\LastVerified$ so there are at most $K' \cdot r(D_\theta)$ such questions.
		\item Questions asked $2k+1$ times during the second verification step when $\Current$
		was pointing to a node outside $P_\theta$.
		This can only happen due to a lie between $\Current$ and $\LastVerified$ so there are at most $K'\cdot(2k+1)$ such questions
		\item Questions that were ignored due to the second verification step,
		when $\Current$ was being backtracked from a node in $P_\theta$. This can only happen if $\Candidate$ is problematic, so there are at most $F'_\theta$ such questions.
		\item Questions asked $2k+1$ during the second verification step when $\Current$
		was pointing to a node in $P$.
		This can only happen if $\Candidate$ is problematic, hence there are at most $F_\theta (2k+1)$ such questions.
	\end{itemize}
\end{proof}

\subsection{Step 3: Culmination of the proof}

In this section we perform calculations to bound the expected number of questions, using Lemma~\ref{lem:alg-running-time}.

\begin{lemma} \label{lem:alg-running-time-exp}
	The expected questions asked on $x_i$ is at most
	\begin{multline*}
	\log(1/\mu_i) + 3 + \EE[K'+1]r(\log(1/\mu_i) + 3) +
	\sum_{d=1}^\infty \binom{r(d)}{\leq k} 2^{-r(d)} (r(d) + 2k + 1) + O(k\EE K').
	\end{multline*}
\end{lemma}
\begin{proof}
	We will use the notation of Lemma~\ref{lem:alg-running-time}.
	Claim~\ref{cla:gilbert-moore-depth} shows that $D_\theta \leq \log(1/\mu_i) + 3$. Since $r(d)$ is monotone nondecreasing, also $r(D_\theta) \leq r(\log(1/\mu_i)+3)$. 
	Let $Z_d$ be the indicator of whether the node of depth $d$ in $P_\theta$ is problematic for $x_i$, for $1 \le d \le \log(1/\mu_i) + 3$.
	Claim~\ref{cla:gilbert-moore-uniform} shows that
	\[
	\EE[Z_d] = \Pr[Z_d = 1]
	\le \binom{r(d)}{\le k} 2^{-r(d)}.
	\]
	Hence,
	\begin{align*}
	\EE_\theta[F_\theta] 
	&= \EE_\theta\left[\sum_{d=1}^{\lfloor \log(1/\mu_i)+3 \rfloor} Z_d \right]
	\le \sum_{d=1}^{\infty} \binom{r(d)}{\leq k} 2^{-r(d)}, \\
	\EE_\theta[F'_\theta] 
	&= \EE_\theta\left[\sum_{d=1}^{\lfloor \log(1/\mu_i)+3 \rfloor} Z_d r(d) \right]
	\le \sum_{d=1}^{\infty} \binom{r(d)}{\leq k} r(d) 2^{-r(d)}. \\
	\end{align*}
	This completes the proof.
\end{proof}

In what follows, we will show that $\sum_{d=1}^\infty \binom{r(d)}{\leq k} r(d) 2^{-r(d)}$ is at most some absolute constant. As $r(d) = \Omega(2k+1)$, this will imply that $\sum_{d=1}^\infty \binom{r(d)}{\leq k} (2k+1) 2^{-r(d)} = O(1)$. Applying Lemma~\ref{lem:alg-running-time-exp} and \eqref{eq:rd-simplified} concludes the proof follows.
We begin with three auxiliary lemmas.

\begin{lemma} \label{lem:binsum}
	For any $n,k \geq 1$, $\binom{n}{\leq k} \leq en^k$.
\end{lemma}
\begin{proof}
	For $\ell \leq k$ we have $\binom{n}{\ell}/n^k \leq (n^\ell/\ell!)/n^k \leq 1/\ell!$. Hence $\binom{n}{\leq k}/n^k \leq \sum_{\ell=0}^k 1/\ell! < e$.
\end{proof}

\begin{lemma} \label{lem:log-eq-aux}
	For any $a,b \ge 1$, $\log (a+b) \le \log a + \log b + 1$. As a consequence, $\ln (a+b) \le \ln a + \ln b + 1$.
\end{lemma}
\begin{proof}
	$\log (a+b) \le \log (2\max(a,b)) \le 1 + \log a + \log b$.
\end{proof}

\begin{lemma} \label{lem:log-eq}
	For all $a,b \ge e$ and all $x \ge b + 4a(\ln a + \ln b)$, it holds that $x \ge a \ln x + b$.
\end{lemma}

\begin{proof}
	Denote $c = 4$. Note that for all $x \ge a$, the function $x - a \ln x - b$ is monotone nondecreasing in $x$ (since its derivative is $1-a/x \geq 0$). 
	Hence, it is sufficient to prove that $x \ge a \ln x + b$ for $x = b + ca(\ln a + \ln b)$, a value which exceeds $a$. Assume, indeed, that $x = b + ca(\ln a + \ln b)$.
	Applying Lemma~\ref{lem:log-eq-aux} twice and using the fact that $c=4$,
	\begin{align*}
	a \ln x
	&= a \ln \left( b + ca(\ln a + \ln b) \right)
	\le a \ln b + a \ln (ca (\ln a + \ln b)) + 1 \\
	&= a \ln b + a \ln ca + a \ln (\ln a + \ln b) + 1 \\
	&\le a \ln b + a \ln a + a \ln c + a \ln \ln a + a \ln \ln b + 2 \\
	&\le 2 a \ln b + 2 a \ln a + a \ln c + 2 a
	\le c a (\ln a + \ln b) = x - b. \qedhere
	\end{align*}
\end{proof}

As stated above, we would like to show that $\sum_{d=1}^\infty \binom{r(d)}{\leq k} r(d) 2^{-r(d)} = O(1)$. In particular, since $\binom{r(d)}{\le k} \le er(d)^k$, it is sufficient to show that $\sum_{d=1}^{\infty} r(d)^{k+1} 2^{-r(d)} = O(1)$. Since $\sum_{d=2}^{\infty} \left(d \ln^2 d\right)^{-1}$ is a convergent series, it is sufficient to show that $r(d)^{k+1} 2^{-r(d)} \le \left((d+1)\ln^2(d+1)\right)^{-1}$. This is equivalent to
\begin{equation} \label{eq:rd-requires}
r(d) \ge \log ((d+1)\ln^2(d+1)) + (k+1) \ln(r(d))/ \ln 2.
\end{equation}
Applying Lemma~\ref{lem:log-eq} (for $k \geq 2$ and large enough $d$) with $a = (k+1)/\ln 2$ and $b = \log \left((d+1)\ln^2(d+1)\right)$, implies that the current definition of $r(d)$ satisfies \eqref{eq:rd-requires}. (We leave the case $k=1$ to the reader.)

\section{Second algorithm}\label{sec:2ndalgproof}

We prove Theorem~\ref{thm:algtwo-one-elem}, from which Theorem~\ref{thm:algtwo-better} follows. 
We start by explaining the main differences between \algone{} and \algtwo{}. The pointer $\Current$ will be defined as before: it simulates a search on the tree, asking one question in every iteration. In this algorithm, $\Current$ simulates $T_\theta$, the finite tree, rather the infinite $T'_\theta$. The pointer $\LastVerified$ is removed. Still, it will be possible to correct lies and $\Current$ will move up the tree whenever a lie is revealed, deleting the recent answers. 
We proceed by giving some definitions:

\begin{definition}
	Two non-root nodes in $T_\theta$ are \emph{matching children} if they are either both right children or both left children. Two non-root nodes are \emph{opposing children} if they are not matching children.
\end{definition}

\begin{definition}
	A node $v$ in the tree $T_\theta$ is a \emph{lie} with respect to an element $x$, if either $v$ is a left child and $x \notin Q(\parent(v))$ or $v$ is a right child and $x \in Q(\parent(v))$.
\end{definition}

In particular, $v$ is a lie if the answer to $Q(\parent(v))$ which causes $\Current$ to move from $\parent(v)$ to $v$ is a lie.

Differently from the first algorithm, a node $v$ will be suspected as a lie only if \emph{all} the descendants in the path to $\Current$ are opposing to $v$ (rather than at most $k$ of them are matching $v$).
In any iteration, we set $\Suspicious$ to be the deepest node in the path from the root to $\Current$ which is an opposing child to $\Current$. All descendants of $\Suspicious$ along this path are opposing children to $\Suspicious$. An action will be taken only if there are $r'(\depth(\Suspicious))$ opposing descendants, where $r' \colon \mathbb{N} \to \mathbb{N}$ is a monotonic nondecreasing function to be defined later. In other words, an action will be taken only if $\depth(\Current) = \depth(\Suspicious) + r'(\depth(\Suspicious))$. If that condition holds, $\Suspicious$ is suspected as a lie, and one sets $\Current \gets \parent(\Suspicious)$. In the next iteration, the question $Q(\parent(\Suspicious))$ will be asked again. We call this action of moving $\Current$ up the tree a \emph{jump back}, or, more specifically, a jump-back atop $\Suspicious$. We add a note: at some iterations to be elaborated later, there will be no $\Suspicious$ node. 
The definition of $r'(j)$ is as follows:
\begin{equation} \label{eq:def-rprime}
r'(j) = \left\lceil \logtwo \left( 2k(j+1)\ln^2(j+1) \right) + e + 4e \left( 1 + \ln \left( \logtwo \left( 2k(j+1)\ln^2(j+1) + e \right) \right) \right) \right\rceil.
\end{equation}
Note that 
\begin{equation} \label{eq:r-tag-nice-def}
r'(j) \leq \logtwo j + C(\olog k + \log \olog j)
\end{equation}
for some constant $C > 0$.

We proceed with a few more definitions. To avoid ambiguity, for any distinct nodes $v \ne v'$, we refer to $Q(v)$ and $Q(v')$ as distinct questions, even if the sets they represent are identical.
\begin{definition}
	Fix a node $v$ and assume that $Q(\parent(v))$ is asked. We say that $v$ is \emph{given as an answer} if either $v$ is a left child and the given answer was ``$x \in Q(\parent(v))$'', or $v$ is a right child and ``$x \notin Q(\parent(v))$'' was given.
\end{definition}
In particular, $v$ is given as an answer if an answer to $Q(\parent(v))$ makes $\Current$ move from $\parent(v)$ to $v$.
Note that if a node is given as an answer $k+1$ times, it is not a lie.
\begin{definition}
	Given a certain point at the execution of the algorithm, we say that a node in $T_\theta$ is \emph{verified} if it was given as an answer at least $k+1$ times before.
\end{definition}
The following claim is obvious.
\begin{claim} \label{cla:at-most-k}
	If $v$ is verified then $v$ is not a lie.
\end{claim}

We proceed by explaining another difference from \algone{}: if a node is $\Suspicious$ and triggers a jump-back, the corresponding question is asked again just \emph{once}, rather than $2k+1$ times. There is no guarantee that the new answer is correct. The simulation of $\Current$ continues according to the new answer. It might happen that the same node will be suspected again and $\Current$ will jump-back again to the same location. Then, the same question will be asked the third time. Note that an answer can be incorrectly suspected, even if no lies are told. This may lead to an infinite loop, where $\Current$ jumps back atop the same node indefinitely. To avoid such a situation, one sets $\Suspicious \gets \None$ if the node which is supposed to be suspicious is verified.

We give a full definition of how $\Suspicious$ is defined: first, if all nodes in the path to $\Current$ (except for the root) are matching children, then $\Suspicious \gets \None$. Otherwise, $\SuspiciousCandidate$ is set as the deepest node in the path from the root to $\Current$ which is an opposing child to $\Current$. If $\SuspiciousCandidate$ is verified then $\Suspicious \gets \None$, otherwise $\Suspicious \gets \SuspiciousCandidate$. The pseudocode for setting $\Suspicious$ appears in the function \textsc{setSuspicious} in Algorithm~\ref{alg:2}.

Before proceeding, we present the following claim, which follows from the discussion on \algone{}.

\begin{claim} \label{cla:lie-opposing}
	If $v$ is a lie, then any descendant of $v$ which is a matching child to $v$ is a lie.
\end{claim}
We add the following definition:
\begin{definition}
	Let $v$ be a node in $T_\theta$. We say that a jump-back deletes $v$ if $v$ was in the path from the root to $\Current$ just before the jump-back, and $v$ is not in the path to $\Current$ right after the jump-back.
\end{definition}

We explain how the algorithm behaves when multiple lies are told. Assume that a lie $v$ was given as an answer, and let $d = \depth(v)$. If no other lies are told, all descendants of $v$ in the path to $\Current$ will be opposing children to $v$. The lie $v$ will be deleted once $\depth(\Current) = d + r'(d)$. However, it may be the case that more lies are told. These are necessarily opposing to $v$, and after they are told, $\Suspicious$ does not equal $v$ and $v$ cannot be deleted. However, $\Suspicious$ will point at the last lie which will be deleted. Then, the other lies will be deleted one after the other. At some point, once all lies which are descendants of $v$ are deleted, $v$ will be pointed by $\Suspicious$ again, and it will finally be deleted once $\depth(\Current) = d + r'(d)$.

Lastly, we explain how the algorithm is terminated and the hidden element is found. In the previous algorithm, this was done once $\LastVerified$ reached a leaf of $T_\theta$. In the absence of $\LastVerified$, we devise a different and more efficient way to verify the correctness of an element. Once $\Current$ reaches a leaf of $T_\theta$, one checks whether the label $e$ of that leaf is indeed the hidden element.
The verification process consists of asking multiple times whether the correct element is $e$. Each question of the form ``element = $e$?'' can be implemented using two comparison questions, asking ``$\preceq e$?'' and ``$\succeq e$?''. For simplicity, assume that the algorithm asks\emph{verification questions} of the form ``$=e$?'' and that the answers are $=$ and $\ne$. The questioner will ask this equality question multiple times until it either gets $k+1$ $=$-answers or until it gets more $\ne$-answers than $=$-answers. If $k+1$ $=$-answers are obtained then the hidden element is $e$ and the search terminates. If more $\ne$-answers than $=$-answers are obtained, the element $e$ is suspected as not being the hidden element. Then, one performs a jump back atop $\Suspicious$ if $\Suspicious \ne \None$ and otherwise one jumps atop $\Current$, setting $\Current \gets \parent(\Current)$. From that point, the next iteration proceeds. The pseudocode of this verification process appears in the function \textsc{verifyObject}, which appears in Algorithm~\ref{alg:2}. We will prove in Lemma~\ref{lem:ub2:vlies} that the total number of verification questions asked during the whole search is $O(k)$. In Lemma~\ref{lem:ver-ret-del}, we will show that if $e$ is not the hidden element then a lie will be deleted upon a return from \textsc{verifyObject}($e$).

The pseudocode of the complete algorithm is presented as Algorithm~\ref{alg:2}. First, an initialization is performed, where $\Current \gets \theroot(T_\theta)$. Then multiple iterations are performed until the element is found. Any iteration consists of the following structure: first, $Q(\Current)$ is asked and $\Current$ is advanced accordingly from parent to child. Then, one checks whether $\Current$ is a leaf of $T_\theta$. If so, the verification process proceeds, and as a result either the algorithm terminates or a jump back is taken. If $\Current$ does not point to a leaf, one checks whether a jump-back atop $\Suspicious$ should be taken and proceeds accordingly.

\begin{algorithm}
	\begin{algorithmic}[1]
		\State{$\theta \gets \mathrm{Uniform}([0,1/2))$}
		\State{$\Current \gets \theroot(T_\theta)$}
		\While{true}
		\If{$x \in Q(\Current)$} \Comment{$Q(\Current)$ is asked}
		\State{$\Current \gets \leftchild(\Current)$}
		\Else
		\State{$\Current \gets \rightchild(\Current)$}
		\EndIf 
		\State{$\Suspicious \gets $ \Call{setSuspicious}{$\Current$}}
		\If{$\Current$ is a leaf of $T_\theta$} \Comment{Checking termination condition}
		\State{$e \gets$ the label of $\Current$}
		\If{\Call{verifyObject}{$e$}}
		\State{\Return $e$}
		\ElsIf{$\Suspicious \ne \None$}
			\State{$\Current \gets \parent(\Suspicious)$}
		\Else
			\State{$\Current \gets \parent(\Current)$}
		\EndIf
		\ElsIf{$\Suspicious \ne \None$ \textbf{ and } $\depth(\Current) = \depth(\Suspicious) + r'(\depth(\Suspicious))$}
		\State{$\Current \gets \parent(\Suspicious)$} \Comment{Jumping-back atop $\Suspicious$}
		\EndIf
		\EndWhile
		\State{\Return label of $\LastVerified$}
		\State{}
		\Function{verifyObject}{$e$}
		\State{Ask the question ``$=e$?'' repeatedly, until either:}
		\State{(1) $k+1$ $=$-answers are obtained. In that case:  \Return true}
		\State{(2) More $\ne$-answers than $=$-answers are obtained. In that case:  \Return false}
		\EndFunction
		\State{}
		\Function{setSuspicious}{$\Current$}
			\If{All nodes in the path from the root (excluding) to $\Current$ (including) are matching children}
				\State{\Return $\None$}
			\Else
				\State{$\SuspiciousCandidate \gets$ the deepest node in the path from the root to $\Current$ which is an opposing child with $\Current$}
				\If{$\SuspiciousCandidate$ is verified}
					\State{\Return $\None$}
				\Else
					\State{\Return $\SuspiciousCandidate$}
				\EndIf
			\EndIf
		\EndFunction
	\end{algorithmic}
	\caption{Resilient-Tree}
	\label{alg:2}
\end{algorithm}

\subsection{Proof}

Fix a hidden element $x_i$, and let $P$ be the path from the root to the leaf labeled $x_i$ in $T_\theta$. A node $v \in P$ of depth $d$ is \emph{problematic} for $x_i$ if all the $r'(d)$ closest descendants of $v$ in $P$ are opposing children with $v$ (in particular, $v$ cannot be problematic if there are less than $r'(d)$ descendants of $v$ in $P$). Note the significance of a problematic node: if $v$ is problematic, then even if no lies are told, $v$ will be $\Suspicious$ and will initiate a jump-back. The following lemmas categorize the different jump-backs taken throughout the algorithm.

\begin{lemma} \label{lem:cur-or-sus}
	Fix some iteration of the algorithm where $\Current$ does not reside in $P$. Then, either $\Suspicious$ is a lie or $\Current$ is a lie. In particular, if $\Suspicious = \None$ then $\Current$ is a lie.
\end{lemma}

\begin{proof}
	Assume there is a lie, and divide into different cases according to $\SuspiciousCandidate$:
	\begin{itemize}
		\item $\SuspiciousCandidate$ is undefined: in that case, all nodes in the path from the root (excluding) to $\Current$ (including) are matching children. As assumed, there exists a lie $v$ in the path to $\Current$. Since $\Current$ is a matching child and a descendant of $v$, Claim~\ref{cla:lie-opposing} implies that $\Current$ is a lie.
		\item $\SuspiciousCandidate$ is not a lie. As assumed, there is a different lie $v$ in the path to $\Current$. We will show that $v$ is an opposing child with $\SuspiciousCandidate$: if $v$ is an ancestor of $\SuspiciousCandidate$, then these nodes are opposing children, from Claim~\ref{cla:lie-opposing}. If $v$ is a descendant of $\SuspiciousCandidate$ it is an opposing child to $v$ by definition of $\SuspiciousCandidate$. Hence, these two nodes are opposing children. Since $\SuspiciousCandidate$ is an opposing child to $\Current$, the nodes $v$ and $\Current$ are matching children, hence, Claim~\ref{cla:lie-opposing} implies that $\Current$ is a lie as well.
		\item If $\SuspiciousCandidate$ is a lie, then $\Suspicious = \SuspiciousCandidate$ and $\Suspicious$ is a lie.
	\end{itemize}
\end{proof}

We categorize jump-backs following a return from \textsc{verifyObject}.
\begin{lemma} \label{lem:ver-ret-del}
	Assume that \textsc{verifyObject}($e$) is called and returns \textbf{false}, resulting in a jump-back. Then, one of the following applies:
	\begin{enumerate}
		\item The hidden element is $e$ and a lie was told in \textsc{verifyObject}.
		\item The hidden element is not $e$ and the jump-back deletes a lie.
	\end{enumerate}
\end{lemma}

\begin{proof}
	If $e$ is the hidden element, a lie has to be told for the call to return \textbf{false}. Assume that $e$ is not the hidden element.
	After the call returns, a jump-back is either taken atop $\Suspicious$ (if it is defined), or atop $\Current$, if $\Suspicious = \None$. In both cases, Lemma~\ref{lem:cur-or-sus} implies that a lie is deleted.
\end{proof}

We categorize jump-backs taken when $\depth(\Current) = \depth(\Suspicious) + r'(\depth(\Suspicious))$.

\begin{lemma} \label{lem:jb-suspicious}
	Assume a jump-back was taken as a result of $\depth(\Current) = \depth(\Suspicious) + r'(\depth(\Suspicious))$. Then one of the following applies:
	\begin{enumerate}
		\item The jump-back deletes a lie. \label{itm:jb1}
		\item The jump-back is atop a problematic node. \label{itm:jb2}
	\end{enumerate}
\end{lemma}

\begin{proof}
	Divide into cases, according to whether $\Current$ resides in $P$:
	\begin{itemize}
		\item If $\Current$ does not reside in $P$, Lemma~\ref{lem:cur-or-sus} implies that either $\Suspicious$ or $\Current$ is a lie. Since both are going to be deleted, a lie is going to be deleted.
		\item If $\Current$ resides in $P$, the definition of a problematic node implies that $\Suspicious$ is problematic and Item~\ref{itm:jb2} applies.
	\end{itemize}
\end{proof}

Next, we prove that there can be at most $k$ jump-backs atop the same $\Suspicious$ node.

\begin{lemma} \label{lem:jb-verified}
	Assume in some point of the algorithm that $k$ jump-backs atop $v$ were taken before. Then, $v$ cannot be labeled as $\Suspicious$.
\end{lemma}

\begin{proof}
	Assume for contradiction that $v$ is labeled as $\Suspicious$. This implies that $v$ was given as an answer at least $k+1$ times: once before each jump-back and once after the last jump-back, which implies that $v$ is verified and contradicts the fact that $v$ is labeled as $\Suspicious$ (since $\Suspicious$ cannot be verified by definition).
\end{proof}

We finalize the categorization of jump-backs with an immediate corollary of Lemma~\ref{lem:ver-ret-del}, Lemma~\ref{lem:jb-suspicious} and Lemma~\ref{lem:jb-verified}.
\begin{corollary} \label{cor:jb-divide}
	Each jump back can be categorizes as one the following:
	\begin{itemize}
		\item A jump-back that either deletes a lie or that is taken as a result of a lie being told in \textsc{verifyObject}. There are at most $K'$ such jump-backs. \label{itm:jb-lie}
		\item A jump-back atop a problematic node $v$ labeled $\Suspicious$. There can be at most $k$ such jump-backs for each node $v$. \label{itm:jb-suspicious}
	\end{itemize}
\end{corollary}

After categorizing the jump-backs, we categorize the different questions asked by the algorithm. Let $M$ be the largest depth of $\Current$ throughout the algorithm. Let $\depth(x_i)$ be the depth of the leaf of $T_\theta$ labeled $x_i$. Let $L$ be the number of nodes problematic for $x_i$ and let $D_1, \dots, D_L$ be the depths of these nodes. Let $V$ be the number of verification questions asked by the algorithm (namely, ``$=e$?'' questions asked in \textsc{verifyObject}).

\begin{lemma} \label{lem:ub2:grouping-q}
	The number of questions asked by the algorithm on element $x_i$ is at most
	\[
	2V + \depth(x_i) + K' (r'(M) + 1) + k \sum_{j=1}^L (r'(D_j)+1).
	\]
\end{lemma}

\begin{proof}
	We say that an answer is deleted by a jump-back if the node which corresponds to this answer is deleted (the node where $\Current$ moves right after the answer is told). Note that the same question can be asked and deleted multiple times, each counted separately. We will count the answers rather than the questions, dividing them into multiple categories:
	\begin{itemize}
		\item Answers to verification questions, amounting to $2V$ answers, since each verification question is implemented using two $\prec$ questions.
		\item Answers not deleted: these correspond to questions in the path from the root to the leaf labeled $x_i$, amounting to $\depth(x_i)$ questions.
		\item Answers deleted as a result of a jump-back categorized as Item~\ref{itm:jb-lie} in Corollary~\ref{cor:jb-divide}. There are at most $K'$ such jump-backs, each deleting at most $r'(M)+1$ answers, to a total of $K'(r'(M)+1)$ deleted answers.
		\item Answers deleted as a result of a jump back categorized as Item~\ref{itm:jb-lie} in Corollary~\ref{cor:jb-divide}. There can be at most $k$ such jump-backs for each problematic node, and a total of $k \sum_{j=1}^L (r'(d_j)+1)$ questions.
	\end{itemize}
\end{proof}

To complete the proof, we bound the different terms in Lemma~\ref{lem:ub2:grouping-q}. We start by bounding the first term.

\begin{lemma} \label{lem:ub2:vlies}
	$V \le 3k+1$.
\end{lemma}

\begin{proof}
	Let $K'_1$ be the total number of lies in the verification questions, and let $K'_2$ be the number of times that the function \textsc{verifyObject} was invoked with a request to verify an incorrect element.
	Divide the jump-backs into different categories:
	\begin{itemize}
		\item Number of questions asked over invocations of \textsc{verifyObject} which ended in a success: note that \textsc{verifyObject} ends with success only once. In that case, $k+1$ $=$-answers are obtained and $K'_{1,1}$ $\ne$-answers are obtained, for some $K'_{1,1} \in \{ 0,1,\dots, k\}$. The total number of questions asked is $k+1 + K'_{1,1}$.
		\item The number of questions asked over invocations of \textsc{verifyObject} which ended in failure, when the object tested was the correct object. In such cases, $m+1$ $\ne$-answers are obtained and $m$ $=$-answers are obtained, for some $m \in \{0,1,\dots,k-1\}$. The number of lies is $m+1$ and the number of questions asked is $2m+1\le 2(m+1)$. The total number of questions asked over all such invocations is at most $2 K'_{1,2}$, where $K'_{1,2}$ is the total number of lies over such invocations.
		\item The number of quesions asked when the answer should be $\ne$: in such cases, $m+1$ $\ne$-answers are obtained and $m$ $=$-answers are obtained, for some $m \in \{0,1,\dots,k\}$. Note that $2m+1$ questions are told where $m$ is the number of lies. Summing over all such invocations, one obtains a bound of $K'_2 + 2 K'_{1,3}$ where $K'_{1,3}$ is the number of lies told in such invocations.
	\end{itemize}
	Summing the over the different categories and noting that $K'_{1,1} + K'_{1,2} + K'_{1,3} = K'_1$, one concludes that the number of verification questions asked is at most $k+1 + 2K'_1 + K'_2$. By Lemma~\ref{lem:ver-ret-del}, each time that the function was invoked in a request to verify an incorrect object was followed by a jump back which caused a lie to be deleted. Hence, $K'_1 + K'_2 \le k$. This concludes the proof.
\end{proof}

The second term in Lemma~\ref{lem:ub2:grouping-q} is bounded by $\log (1/\mu_i) + 3$ using Claim~\ref{cla:gilbert-moore-depth}. We proceed by bounding the third term.
To bound $M$, we start with the following auxiliary lemma:

\begin{lemma} \label{lem:log-seq}
	Let $q_0,a \ge 1$. Define $q_i = q_{i-1} + a \log q_{i-1} + a$, for all $i > 0$. Then, it holds that $q_i \le q_0 + C a i (\log q_0 + \log a +\log i + 1)$, for the constant $C = 8$.
\end{lemma}

\begin{proof}
	We bound $q_i$ by induction on $i$. For $i = 0$ and $i = 1$ the statement is trivial. For the induction step, fix some $i \ge 1$. Using the induction hypothesis, Lemma~\ref{lem:log-eq-aux}, and the inequality $\logtwo x \le x$, it holds that
	\begin{align*}
	\frac{q_{i+1} - q_i}{a} - 1 
	&= \log q_i
	\le \log \left( q_0 + C a i (\log q_0 + \log a + \log i + 1) \right) \\
	&\le \log q_0 + \log (C a i (\log q_0 + \log a + \log i + 1)) + 1 \\
	&= \log q_0 + \log C + \log a + \log i + \log(\log q_0 + \log a + \log i + 1) + 1 \\
	&\le \log q_0 + \log C + \log a + \log i + \log q_0 + \log a + \log i + 1 + 1\\
	&= 2\log q_0 + \log C + 2 \log a + 2\log i + 2.
	\end{align*}
	The result follows by substituting $C = 8$ and applying the induction hypothesis.
\end{proof}

The next lemma bounds $M$:

\begin{lemma} \label{lem:bnd-M}
	It always holds that $M \le \olog(1/\mu_i) C k \olog k(\olog k + \log \olog (1/\mu_i))$.
\end{lemma}

\begin{proof}
	Define $C' \ge 1$ to be a sufficiently large constant such that for all $j \ge 1$: $r'(j) \ge a \log j + a$ for $a = C' \olog k$. It follows from Eq.~\eqref{eq:r-tag-nice-def} that such a value of $C'$ exists.
	Define the infinite sequence $q_0, q_1, q_2,  \dots$ as follows: $q_0 = \log (1/\mu(x)) + 3$ and $q_{j+1} = q_j + a \log q_j + a$ for the value of $a$ defined above.
	
	Claim~\ref{cla:gilbert-moore-depth} implies that $q_0$ bounds form above the maximal depth of $\Current$ when there are no lies. Additionally, note that by induction, $q_{j+1} \ge q_j + r'(q_j)$ bounds from above the maximal depth of $\Current$ when there are at most $j$ lies.
	Applying Lemma~\ref{lem:log-seq}, one obtains the desired result.
\end{proof}

As a corollary, we bound $r'(M)$, which corresponds to the third term in Lemma~\ref{lem:ub2:grouping-q}.

\begin{corollary} \label{cor:bnd-rprime-M}
	\[ 
	r'(M) \le \log \olog (1/\mu_i) + O(\log \log \log (1/\mu_i) + \olog k).
	\]
\end{corollary}

\begin{proof}
	This follows immediately from Lemma~\ref{lem:bnd-M} and Eq.~\eqref{eq:r-tag-nice-def}.
\end{proof}

Lastly, we bound the last term in Lemma~\ref{lem:ub2:grouping-q}.

\begin{lemma}
	\[
	\EE\left[k \sum_{j=1}^L (r'(D_j)+1) \right] \le C
	\]
	for some universal constant $C > 0$, where the expectation is over the random choice of $0 \le \theta < 1/2$.
\end{lemma}

\begin{proof}
	For $m \in \mathbb{N}$, let $Z_m$ be the indicator of whether the node at depth $m$ in $P$ is problematic for $x_i$ (in particular, $Z_m = 0$ if $\depth(x_i) < m$). From Lemma~\ref{cla:gilbert-moore-uniform}, it holds that $\EE[Z_m] = \Pr[Z_m = 1] \le 2^{-r'(m)}$. In particular,
	\begin{equation}
	\label{eq:ub2:bndsum}
	\EE\left[\sum_{j=1}^L k(r'(D_j)+1)\right]
	= \EE\left[\sum_{m=1}^\infty Z_m k(r'(m)+1) \right]
	\le \sum_{m=1}^{\infty} 2^{-r'(m)}k(r'(m)+1).
	\end{equation}
	For simplicity, we will bound $\sum_{m=1}^{\infty} 2^{-r'(m)}k r'(m)$ which is within a constant factor of the right hand side of \eqref{eq:ub2:bndsum}.
	Since $\sum_{m=2}^\infty m^{-1} \log^{-2}(m)$ is a convergent series, it is sufficient that $2^{-r'(m)} r'(m) \le \left(2k(m+1)\ln^2(m+1)\right)^{-1}$. This is equivalent to $r'(m) \ge \logtwo r'(m) + \logtwo\left( 2k(m+1)\ln^2(m+1) \right)$, hence it is sufficient to require
	$r'(m) \ge e \ln r'(m) + \logtwo\left( 2k(m+1)\ln^2(m+1) \right) + e$.
	Lemma~\ref{lem:log-eq} implies that this inequality holds for $r'(m)$ (defined in \eqref{eq:def-rprime}) as required.
\end{proof}

The proof follows from Lemma~\ref{lem:ub2:grouping-q}, Lemma~\ref{lem:ub2:vlies}, Claim~\ref{cla:gilbert-moore-depth}, Corollary~\ref{cor:bnd-rprime-M} and Lemma~\ref{cor:bnd-rprime-M}.

\section{Sorting} \label{sec:sort}

We present the proof of Theorem~\ref{thm:sort-ent}.

\begin{proof}[Proof of Theorem~\ref{thm:sort-ent}]
	Let $\Pi \colon [k] \to [k]$ be random variable defining the correct ordering over the elements, namely, $x_{\Pi^{-1}(1)} < x_{\Pi^{-1}(2)} < \cdots < x_{\Pi^{-1}(n)}$, and 
	let $\Pi_\ell \colon [\ell] \to [\ell]$ be the permutation defining the correct ordering between $x_1, \dots, x_\ell$, namely, $x_{\Pi_\ell^{-1}(1)} < x_{\Pi_\ell^{-1}(2)} < \cdots < x_{\Pi_\ell^{-1}(\ell)}$. In other words, for all $1 \le i,j \le \ell$, $x_i < x_j$ if and only if $\Pi_\ell(i) < \Pi_\ell(j)$.
	
	Our sorting procedure proceeds in iterations, finding the correct ordering between $x_1, \dots, x_\ell$ by the end the $\ell$'th iteration, for $\ell = 1,\dots, n$. In other words, it finds $\Pi_\ell$ on iteration $\ell$. Given $\Pi_{\ell - 1}$, one only has to find $\Pi_\ell(\ell)$. This can be implemented using comparison questions: the question ``$\Pi_\ell(i) \le r$?'' is equivalent to ``$x_i < \Pi_{\ell-1}^{-1}(r)$?''.
	
	The resulting algorithm is simple: in each iteration $\ell = 1, \dots, n$, find $\Pi_\ell(\ell)$ by invoking \algtwo{} with the distribution $\Pi_\ell(\ell)\mid \Pi_{\ell-1}$.
	
	We will bound the expected number of questions asked by this algorithm. Fix some permutation $\pi$, set 
	$p_\pi = \Pr[\Pi = \pi]$, and set $p_{\pi,\ell}$ as the probability that $\Pi_\ell$ agress with $\pi$ conditioned on $\Pi_{\ell-1}$ agreeing with $\pi$. Define $k'_\ell$ as the expected number of lies on round $\ell$. It follows from Theorem~\ref{thm:algtwo-better} that the expected number of questions asked in iteration $\ell$ is at most 
	\[
	\log\frac{1}{p_{\pi,\ell}} + O\left(k'_\ell \log \olog \frac{1}{p_{\pi,\ell}} + k'_\ell \log k + k\right).
	\]
	Summing over $\ell = 1, \dots, n$, one obtains:
	\[
	\log\frac{1}{\prod_{\ell=1}^n p_{\pi,\ell}} + O\left(\sum_{\ell=1}^n \left(k'_\ell \log \olog \frac{1}{p_{\pi}} + k'_\ell \log k + k\right) \right)
	\le \log\frac{1}{p_{\pi}} + O \left(k \log \olog \frac{1}{p_{\pi}} + k \log k + k n\right).
	\]
	Taking expectation over $\pi \sim \Pi$ and applying Jensen's inequality with $x \mapsto \log x$, one obtains a bound of
	\[
	H(\Pi) + O(k \log H(\pi) + k \log k + kn) \le H(\Pi) + O(k \log k + kn).
	\]
\end{proof}

\begin{thebibliography}{10}

\bibitem{Aigner}
Martin Aigner.
\newblock Finding the maximum and minimum.
\newblock {\em Discrete Applied Mathematics}, 74:1--12, 1997.

\bibitem{AslamDhagat}
Javed~A. Aslam and Aditi Dhagat.
\newblock Searching in the presence of linearly bounded errors.
\newblock In {\em Proceedings of the 23rd annual Symposium on Theory of
  Computing (STOC'91)}, pages 486--493, 1991.

\bibitem{Bagchi}
A.~Bagchi.
\newblock On sorting in the presence of erroneous information.
\newblock {\em Information Processing Letters}, 43:213--215, 1992.

\bibitem{Berlekamp}
Elwyn~R. Berlekamp.
\newblock {\em Block coding for the binary symmetric channel with noiseless,
  delayless feedback}, pages 61--85.
\newblock Wiley, New York, 1968.

\bibitem{BK93}
R.~Sean Borgstrom and S.~Rao Kosaraju.
\newblock Comparison based search in the presence of errors.
\newblock In {\em Proceedings of the 25th annual symposium on theory of
  computing (STOC'93), year = {1993}, pages = {130--136}}.

\bibitem{BM09}
Mark Braverman and Elchanan Mossel.
\newblock Noisy sorting without resampling.
\newblock In {\em SODA}, pages 268--276, 2008.

\bibitem{dfgm}
Yuval Dagan, Yuval Filmus, Ariel Gabizon, and Shay Moran.
\newblock Twenty (simple) questions.
\newblock In {\em 49th ACM Symposium on Theory of Computing (STOC 2017)}, 2017.

\bibitem{Fano}
Robert~Mario Fano.
\newblock The transmission of information.
\newblock Technical Report~65, Research Laboratory of Electronics at MIT,
  Cambridge (Mass.), USA, 1949.

\bibitem{FPRU94}
Uriel Feige, David Peleg, Prabhakar Raghavan, and Eli Upfal.
\newblock Computing with noisy information.
\newblock {\em SIAM Journal on Computing}, 23:1001--1018, 1994.

\bibitem{gelles2017coding}
Ran Gelles et~al.
\newblock Coding for interactive communication: A survey.
\newblock {\em Foundations and Trends{\textregistered} in Theoretical Computer
  Science}, 13(1--2):1--157, 2017.

\bibitem{GilbertMoore}
E.~N. Gilbert and E.~F. Moore.
\newblock Variable-length binary encodings.
\newblock {\em Bell System Technical Journal}, 38:933--967, 1959.

\bibitem{haeupler2014interactive}
Bernhard Haeupler.
\newblock Interactive channel capacity revisited.
\newblock In {\em Foundations of Computer Science (FOCS), 2014 IEEE 55th Annual
  Symposium on}, pages 226--235. IEEE, 2014.

\bibitem{kol2013interactive}
Gillat Kol and Ran Raz.
\newblock Interactive channel capacity.
\newblock In {\em Proceedings of the forty-fifth annual ACM symposium on Theory
  of computing}, pages 715--724. ACM, 2013.

\bibitem{LRG91}
K.B. Lakshmanan, B.~Ravikumar, and K.~Ganesan.
\newblock Coping with erroneous information while sorting.
\newblock {\em IEEE Transactions on Computers}, 40:1081--1084, 1991.

\bibitem{Long}
Philip~M. Long.
\newblock Sorting and searching with a faulty comparison oracle.
\newblock Technical Report UCSC--CRL--92--15, University of California at Santa
  Cruz, November 1992.

\bibitem{MY13}
Shay Moran and Amir Yehudayoff.
\newblock A note on average-case sorting.
\newblock {\em Order}, 33(1):23--28, 2016.

\bibitem{Pelc87}
Andrzej Pelc.
\newblock Coding with bounded error fraction.
\newblock {\em Ars Combinatorica}, 42:17--22, 1987.

\bibitem{Pelc}
Andrzej Pelc.
\newblock Searching games with errors---fifty years of coping with liars.
\newblock {\em Theoretical Computer Science}, 270:71--109, 2002.

\bibitem{Renyi}
Alfr\'ed R\'enyi.
\newblock On a problem of information theory.
\newblock {\em MTA Mat. Kut. Int. Kozl.}, 6B:505--516, 1961.

\bibitem{Rivest}
Ronald~L. Rivest, Albert~R. Meyer, Daniel~J. Kleitman, and Karl Winklmann.
\newblock Coping with errors in binary search procedures.
\newblock {\em Journal of Computer and System Sciences}, 20:396--404, 1980.

\bibitem{schulman1996coding}
Leonard~J Schulman.
\newblock Coding for interactive communication.
\newblock {\em IEEE transactions on information theory}, 42(6):1745--1756,
  1996.

\bibitem{Shannon}
Claude~Elwood Shannon.
\newblock A mathematical theory of communication.
\newblock {\em Bell System Technical Journal}, 27:379--423, 1948.

\bibitem{SpencerWinkler}
Joel Spencer and Peter Winkler.
\newblock Three thresholds for a liar.
\newblock {\em Combin. Probab. Comput.}, 1(1):81--93, 1992.

\bibitem{Ulam}
Stanislav~M. Ulam.
\newblock {\em Adventures of a mathematician}.
\newblock Scribner's, New York, 1976.

\end{thebibliography}

\end{document}